\documentclass[hidelinks,11pt]{article}
\usepackage{JASA_manu}

\usepackage{natbib,color,hyperref,bm,amsmath,amsthm,amssymb,graphicx,float,booktabs,subcaption,multirow,threeparttable,textcomp,times}
\usepackage{tikz}
\usetikzlibrary{positioning,shapes.geometric}
\usepackage{framed,enumitem}

\theoremstyle{plain}
\newtheorem{theorem}{Theorem}
\newtheorem{lemma}{Lemma}
\newtheorem{assumption}{Assumption}
\newtheorem{proposition}{Proposition}
\newtheorem{example}{Example}

\def\pr{{f}}

\def\tpr{\tilde {f}}

\def\A{{\mathcal  A}}
\def\B{{\mathcal  B}}
\def\S{{\mathcal  S}}
\def\C{{\mathcal C}}
\def\i{{\rm  i}}
\def\med{{\rm median}}

\def\T{{ \mathrm{\scriptscriptstyle T} }}

\makeatletter
\newcommand*{\ind}{%
	\mathbin{%
		\mathpalette{\@ind}{}%
	}%
}
\newcommand*{\nind}{%
	\mathbin{% % The final symbol is a binary math operator
		\mathpalette{\@ind}{\not}% \mathpalette helps for the adaptation
	}%
}
\newcommand*{\@ind}[2]{
	\sbox0{$#1\perp\m@th$}% box 0 contains \perp symbol
	\sbox2{$#1=$}% box 2 for the height of =
	\sbox4{$#1\vcenter{}$}% box 4 for the height of the math axis
	\rlap{\copy0}% first \perp
	\dimen@=\dimexpr\ht2-\ht4-.2pt\relax
	\kern\dimen@
	{#2}%
	\kern\dimen@
	\copy0 % second \perp
} 
\makeatother

\usepackage{titling}

\begin{document}

%\vspace{-2cm}

\title{\bf Identifying effects of multiple treatments in the presence of   unmeasured confounding}
\date{}
\author{
Wang Miao \thanks{Department of Probability and Statistics, Peking University,  Beijing, PRC; mwfy@pku.edu.cn}, 
Wenjie Hu \thanks{Department of Probability and Statistics, Peking University, Beijing, PRC; huwenjie@pku.edu.cn}, 
Elizabeth L. Ogburn \thanks{Department of Biostatistics, Johns Hopkins Bloomberg School of Public Health, Baltimore, MD, USA; eogburn@jhsph.edu}, and  
Xiaohua Zhou \thanks{
Department of Biostatistics and Beijing International Center for Mathematical Research, 
 \protect\\[-5pt] \hspace*{2em}  Peking University, Beijing, PRC;   azhou@math.pku.edu.cn}
 }

\maketitle

%\pagenumbering{gobble} 

\vspace{-2cm}

\begin{center}
\textbf{Abstract}
\end{center}

Identification  of treatment effects in the presence of unmeasured confounding is a persistent problem in the social, biological, and medical sciences.
The problem of unmeasured confounding in settings with multiple treatments is most common in statistical genetics and bioinformatics settings, where researchers have developed many successful statistical strategies without engaging deeply with the causal aspects of the problem. Recently there have been a number of attempts to bridge the gap between these statistical approaches and causal inference, but these attempts have either been shown to be flawed or have relied on fully parametric assumptions. In this paper, we propose two strategies for identifying and estimating causal effects of multiple treatments in the presence of unmeasured confounding. The \emph{auxiliary variables} approach leverages   variables that are not causally associated with the outcome; in the case of a univariate confounder, our method only requires one auxiliary variable, unlike existing instrumental variable methods that would require as many instruments as there are treatments. 
An alternative \emph{null treatments} approach relies on the assumption that at least half of the confounded treatments have no causal effect on the outcome, but does not require a priori knowledge of which treatments are null. Our identification strategies do not impose  parametric  assumptions on the outcome model and do not rest on estimation of the confounder. 
This paper extends and generalizes existing work on unmeasured confounding with a single treatment and models commonly used in bioinformatics.

\vspace*{.3in}

\vspace*{.3in}

\newpage
% \pagenumbering{arabic}

\section{Introduction} 

Identification  of treatment effects in the presence of unmeasured confounding is a persistent problem in the social, biological, and medical sciences, where in many settings it is difficult to  collect data on all possible treatment-outcome confounders. 
Identification means that the treatment effect  of interest is uniquely determined from the joint distribution of observed variables. Without identification, statistical inference may  be misleading  and is of limited interest. 
Most of the work on unmeasured confounding by causal inference researchers focuses on settings with a single treatment, and either harnesses auxiliary variables (e.g. instruments, negative controls, or confounder proxies) to achieve point identification of causal effects, or relies on sensitivity analyses or on weak assumptions to derive bounds for the effects of interest. 
A large body of work from statistical genetics and computational biology is concerned with  multiple treatments--for example GWAS (Genome-Wide Association Studies) with  confounding by population structure and computational biology applications with confounding by batch effects. 
Recently, there have been a few attempts to put these approaches on solid theoretical footing and to bridge the gap between these statistical approaches and causal inference.
However, these attempts either rely themselves on strong parametric models that circumvent  the underlying causal structure,  or have been shown to be  flawed.
In this paper, we propose two novel strategies for identifying causal effects of multiple treatments in the presence of unmeasured  confounding without placing any parametric restrictions on the  outcome model.
This paper  generalizes existing work on unmeasured confounding with a single treatment to the multi-treatment  setting,   and resolves challenges that have undermined previous proposals for dealing with multi-treatment unmeasured confounding. 

\subsection{Related work}

For a single treatment, a variety of methods have been developed to test, adjust for, and eliminate unmeasured confounding bias. Sensitivity analysis \citep{cornfield1959,rosenbaum1983assessing,ding2014generalized} and bounding  \citep{manski1990nonparametric, balke1997bounds,richardson2014ace} are used to evaluate the robustness of causal inference to unmeasured confounding. For point identification of the treatment effect, the instrumental variable (IV) is  an influential tool used in biomedical, epidemiological, and socioeconomic studies \citep{wright1928tariff,goldberger1972,robins1994correcting,angrist1996identification,didelez2007mendelian,small2017instrumental}. 
Recently, \cite{miao2018proxy},   \cite{shi2020multiply}, \cite{tchetgen2020introduction},  \cite{lipsitch2010negative},  \cite{kuroki2014measurement},  \cite{ogburn2012nondifferential},  \cite{flanders2017new}, and \cite{wang2017confounder} demonstrate the potential of using confounder proxy variables and negative controls for adjustment of confounding bias. For an overview of recent work in the single treatment setting, see \cite{tchetgen2020introduction} and \cite{wang2018bounded}.

Similar methods can sometimes be used in  settings with multiple treatments,  simply treating them as a single vector-valued treatment.
These approaches   allow  for unrestricted correlations   among the  treatments. However,   if, as is typically the case in GWAS and computational biology settings, correlations among treatments contains useful information about the confounding, these methods cannot leverage the information. 
Latent variable methods leveraging the multi-treatment correlation structure  have been used to estimate and control for unmeasured confounders in biological applications since the early 2000s \citep{alter2000singular, price2006principal,  leek2007capturing,     friguet2009factor, gagnon2012using,luo2019batch}. 
Recently, a few authors have attempted  to elucidate the causal structure underlying these statistical procedures   and to establish  rigorous   theoretical  guarantees for identification,  using fully parametric models.
\cite{wang2017confounder} propose confounding adjustment  approaches for  the effects of a treatment on multiple outcomes under  a linear  factor model;
by  reversing the labeling of the outcome and   treatments, their approaches can test but not  identify the effects of multiple treatments on the outcome.  
\cite{kong2019multi}  consider   a binary outcome with a univariate confounder and prove identification under a linear factor model for multiple treatments and a parametric  outcome model  via a meticulous analysis  of the link distribution; but their approach cannot generalize to the multivariate confounder setting as we illustrate with a counterexample in the supplement.
\cite{grimmer2020ive}, \cite{cevid2020spectral}, \cite{guo2020doubly}, and \cite{chernozhukov2017lava} consider linear outcome models with high-dimensional treatments that are confounded or mismeasured; 
in this case,  identification  is implied by the fact that confounding on each  treatment vanishes as the number of treatments goes to infinity. 
In contrast,  we take a fundamentally causal approach to confounding and to identification of treatment effects by allowing the outcome model to be  unrestricted, the treatment-confounder distribution to lie in a more general, though not unrestricted, class of models,  the number of treatments to be finite, and confounding to not vanish.

Most notably, \cite{wang2019jasa} provide an intuitive  justification for using latent variable methods in general multi-treatment unmeasured confounding settings; they call the justification and resulting method the ``deconfounder." 
Their approach uses  a factor model assuming that  treatments are independent conditional on the confounder   to estimate the  confounder, and  the confounder estimate is used for adjustment of   bias. 
However,   as demonstrated in a counterexample by \cite{damour2019multi} and discussed by \cite{ogburn2019comment,ogburn2020counterexamples} and \cite{imai2019discussion}, 
identification  is not guaranteed for the deconfounder,  i.e., 
the treatment effects  can not be uniquely determined from the observed data  even with an infinite number of
data samples. 
Additionally, an infinite number of treatments are required for consistent estimation of the confounder, complicating finite sample inference and undermining positivity.
For refinements and discussions of the deconfounder approach, 
see  \cite{wang2020towards}, \cite{damour2019multi}, \cite{ogburn2020counterexamples}, \cite{grimmer2020ive}, and the commentaries
\citep{damour2019comment,ogburn2019comment,imai2019discussion,athey2019comment} published alongside \cite{wang2019jasa}.
 \cite{damour2019multi}  suggests the proximal inference  and \cite{imai2019discussion} consider the conventional instrumental variable approach to facilitate identification. 
However, if correlations among the multiple treatments are indicative of confounding, as the deconfounder approach assumes, neither of these two methods makes use of that correlation. 
Moreover,  their extension to the  multi-treatment setting is complicated by  the fact that the  proximal inference 
requires    confounder proxies to be  causally uncorrelated with  any  of the treatments and the instrumental variable approach requires at least as many instrumental variables as there are treatments.

\subsection{Contribution}
In Section \ref{sec:id}, we review the challenges for identifying multi-treatment effects in the presence of unmeasured confounding.
In Sections 3 and 4,  we propose two novel approaches for the identification of causal effects of multiple treatments with unmeasured confounding: an \emph{auxiliary variables} approach and a \emph{null treatments} approach. 
Both approaches rely on two assumptions restricting the joint distribution of the unmeasured confounder and treatments. The first assumption is that the  joint treatment-confounder distribution lies in a class of models that satisfy a particular \emph{equivalence} property that is known to hold for many commonly-used models, e.g., many types of factor and mixture models. 
This assumption can accommodate  other treatment-confounder models,  such as mixture models,  in addition to   the  factor models considered   by \cite{wang2019jasa}, \cite{kong2019multi},   \cite{wang2017confounder}, and \cite{grimmer2020ive}. 
The second assumption is that the treatment-confounder distribution satisfies a \emph{completeness} condition that is standard in nonparametric identification problems.
In addition to these two assumptions, the auxiliary variables approach leverages an   auxiliary variable  that   does  not directly affect the outcome to identify treatment effects, such as an IV or confounder proxy.
In the presence of a univariate confounder, identification can be achieved with our approach even if only one auxiliary variable is available and if it is associated with only one confounded treatment. 
In contrast,     IV approaches     require  as many instrumental variables as there are treatments and that all confounded treatments must be associated with the instrumental variables. 
The null treatments approach does not require  any auxiliary variables,  but instead rests on the assumption   that  at least half of the confounded treatments are null, without requiring knowledge of   which are active and which are null. 
In these two approaches,  identification is  achieved without imposing  parametric  assumptions  on the outcome model, although the   joint treatment-confounder distribution is restricted by the equivalence and the completeness assumptions.
Identification does not rest on estimation of the unmeasured confounder,
and thus  works with a finite number of treatments and does not run afoul of positivity. 
In the absence of auxiliary variables and if the null treatments assumption fails to hold, 
our method still constitutes a valid test of the null hypothesis of no joint treatment effect.
Because identification in both approaches  requires solving an integral equation, an explicit identification formula is not available for unrestricted  outcome models. However, we describe some estimation strategies in Section \ref{sec:estimation}. In simulations in Section \ref{sec:sims}, the proposed approaches perform well with little bias and appropriate coverage rates. In a data example about mouse obesity in Section \ref{sec:application}, we apply the approaches to detect genes possibly causing mouse obesity, which reinforces previous findings by taking      unmeasured confounding into account. 
Section \ref{sec:conclusion} concludes with a brief mention of some potential extensions of our approaches. Proofs and further discussions are relegated to the supplement.

\section{Preliminaries and challenges to identification} \label{sec:id}

Throughout the paper, we let $X=(X_1,\ldots, X_p)^\T$ denote a vector of $p$ treatments and $Y$ an outcome. We are interested in
the effects of $X$ on $Y$, which may be confounded by a vector of $q$ unobserved covariates $U$. 
The dimension of the confounder, $q$, is assumed to be known a priori; 
for choice of $q$ in practice see  illustrations in Sections 6--7 and  the  discussion in Section 8.
For notational convenience, we suppress observed covariates, and all conditions and results can be viewed as conditioning on them. 
Hereafter, we use $\Sigma_A$ to denote the covariance matrix of a random vector $A$.
We use $\pr$ to denote  a generic probability density or mass function and   $\pr(A=a\mid B=b)$  the conditional density/mass of $A$ given $B$ evaluated at $(A=a,B=b)$, and  write $\pr(a\mid b)$  for simplicity. 
Vectors are assumed to be column vectors unless explicitly transposed.
We refer to $\pr(x,u)$ as the treatment-confounder distribution and $\pr(y\mid u,x)$ the outcome model.

Let $Y (x)$ denote the potential outcome that would have  been observed had the treatment $X$ been set to $x$. 
Treatment effects are defined by  contrasts of potential outcomes between different treatment conditions,
and thus we focus on identification of  $\pr\{Y(x)\}$. 
We say that $\pr\{Y(x)\}$ is identified if and only if it is  uniquely determined by  the  joint distribution of observed variables.

Throughout we make three standard identifying assumptions.
\begin{assumption}\label{assump:ign}
	\begin{enumerate}
	
		\item[(i)] Consistency: When $X=x$, $Y=Y(x)$;  
		
		\item[(ii)]Ignorability:  $Y(x)\ind  X\mid U$;

		\item [(iii)] Positivity: $0<\pr(X=x\mid U=u)<1$ for all $(x,u)$.

	\end{enumerate}
\end{assumption}
Consistency states that the observed outcome is a realization of the potential outcome under the treatment actually received. 
Ignorability, also called ``exchangeability,"  ensures that treatment assignments are effectively randomized conditional on $U$ and implies that $U$ suffices to control for all confounding. Positivity, also called ``overlap," ensures that for all values of $U$ all treatment values have positive probability.

If we were able to observe the confounder $U$,  Assumption \ref{assump:ign} would permit fully nonparametric identification of $\pr\{Y(x)\}$  by the back-door formula \citep{pearl1995causal}  or the g-formula \citep{robins1986new}, 
\begin{eqnarray}\label{eq:effect}
\pr\{Y(x)=y\} = \int_u \pr(y\mid u, x)\pr(u) du. 
\end{eqnarray}
%See \cite{rosenbaum1983central} and \cite{hernanbook} for an overview of nonparametric identification of treatment effects under these standard assumptions when $U$ is observed.

But when $U$ is not observed,  all information contained in the observed data is  captured by $\pr(y,x)$, from which one cannot uniquely determine the joint distribution $\pr(y,x,u)$.
To be specific, one has to solve for $\pr(x,u)$  and $\pr (y \mid u, x)$  from 
\begin{eqnarray}
\pr(x)&=&\int_u \pr(x,u) du,\label{eq:factor0}\\
\pr(y\mid x)&=&\int_u \pr(y\mid u,x)\pr(u\mid x) du. \label{eq:int0}
\end{eqnarray}
However, $\pr(x,u)$ cannot be uniquely determined from \eqref{eq:factor0}, even if, as is  common practice, a  factor model is imposed on $\pr(x,u)$;
see \cite{damour2019multi} for a counterexample.
Furthermore, even if $\pr(x,u)$ is known, the outcome model $\pr(y\mid u,x)$ cannot be identified; 
\cite{damour2019multi} points out that lack of identification of $\pr(y \mid  u, x)$ is due to the unknown copula of $\pr(y \mid  x)$ and $\pr(u \mid  x)$. 
Here we note that identifying $\pr(y \mid  u,x)$ given $\pr(u\mid x)$ is equivalent to solving the integral equation \eqref{eq:int0}, 
but the solution   is not unique without extra assumptions. 
As a result, one cannot identify the true joint distribution $\pr(y,x,u)$  that is  essential for the g-formula  \eqref{eq:effect}. 
We  call a joint  distribution  $\tpr(y,x,u)$ \emph{admissible} if it  conforms to the observed data distribution $\pr(y,x)$, i.e.,  $\pr(y,x)=\int_u \tpr(y,x,u)du$.
The counterexample by \cite{damour2019multi} also shows that   different  admissible joint distributions  result in different potential outcome distributions, i.e., the potential outcome distribution is not identified without additional assumptions.
Some previous approaches have estimated $U$ directly with a deterministic function of $X$,  
but this controverts the positivity assumption and   requires an infinite number of treatments in order to consistently estimate $U$. 
Furthermore,  \cite{grimmer2020ive} show that  in these settings  the effect of $X$ on $Y$ is asymptotically unconfounded 
and  a naive regression of $Y$ on $X$  weakly  dominates these more involved approaches.

\section{Identification with auxiliary variables}\label{sec:auxil}

\subsection{The auxiliary variables assumption}
Suppose we have  available a vector of auxiliary variables $Z$,  then the observed data distribution is captured by $\pr(x,y,z)$, from which we aim to identify the potential outcome distribution $\pr\{Y(x)\}$.
Let $\pr(x,u\mid z;\alpha)$ denote a model for the treatment-confounder distribution  indexed by a possibly infinite-dimensional parameter $\alpha$,  and $\pr(x\mid z;\alpha)$ the resulting marginal distribution.
Given $\pr(x\mid z;\alpha)$, we let $\pr(x,u\mid z;\tilde \alpha)$ denote an arbitrary admissible joint distribution such that $\pr(x\mid z;\alpha)=\int_u \pr(x,u\mid z;\tilde \alpha)du$, and write  $\tpr(x,u\mid z)=\pr(x,u\mid z;\tilde \alpha)$ for short.
Our identification strategy rests on the following  assumption.

\begin{assumption}\label{assump:auxil}
	\begin{enumerate}
		\item[(i)] Exclusion restriction: $Z\ind Y\mid (X,U)$;  
		\item[(ii)]Equivalence:   for any $\alpha$, any    $ \tilde \pr(x, u\mid z)$  that solves  $\pr(x\mid z;\alpha)=\int_u \tpr(x,u\mid z) du$  can be written as $ \tilde\pr(x, u\mid z) =  \pr\{X=x, V(U)=u\mid z; \alpha\}$ for some   invertible but not necessarily known  function $V$;
\item [(iii)] Completeness: for any $\alpha$, $\pr(u\mid x,z;\alpha)$ is complete in $z$,  i.e.,  for any  fixed $x$ and square-integrable function $g$, $E\{g(U)\mid X=x,Z;\alpha\}=0$ almost surely if and only if $g(U)=0$ almost surely.
	\end{enumerate}
\end{assumption}
The exclusion restriction characterizing   auxiliary variables  is the same condition invoked for  the treatment-inducing  confounder
proxies   by \cite{miao2018proxy} and \cite{tchetgen2020introduction}; 
it  rules out the existence of a direct causal association between  the auxiliary variable and the outcome.
It is satisfied by   instrumental variables and   confounder proxies or negative controls.
Figure \ref{fig:dag} includes two directed acyclic graph (DAG) examples that satisfy the assumption.
In Section 3.3 we will illustrate the difference between our use of  auxiliary variables and previous proposals.
\begin{figure}[H]
\centering
\begin{minipage}{0.45\textwidth}
	\centering
	\begin{tikzpicture}[scale=1,
	->,
	shorten >=2pt,
	>=stealth,
	node distance=2cm,
	pil/.style={
		->,
		thick,
		shorten =2pt,}
	]
	\node (X) at (0,0) {$X$};
	\node (Y) at (3,0) {$Y$};	
	\node (U) at (1.5, 1.5) {$U$};
	\node (Z) at (-1.5,1.5)  {$Z$};
	\foreach \from/\to in {U/X,U/Y, Z/X, X/Y}
	\draw (\from) -- (\to);
	% \draw [-,dashed] (C) to (R);
	% \draw [->,dashed] (C) to (X);
	% \draw [->] (Z) to [out=-30,in=-150] (Y);  
	% \draw [->] (Z) to [out=-30,in=-150] (X2);  
	\end{tikzpicture}
		\subcaption{An instrumental variable}
\end{minipage}
\begin{minipage}{0.45\textwidth}
	\centering
\begin{tikzpicture}[scale=1,
	->,
	shorten >=2pt,
	>=stealth,
	node distance=2cm,
	pil/.style={
		->,
		thick,
		shorten =2pt,}
	]
	\node (X) at (0,0) {$X$};
	\node (Y) at (3,0) {$Y$};	
	\node (U) at (1.5, 1.5) {$U$};
	\node (Z) at (-1.5,1.5) {$Z$};
	\foreach \from/\to in {U/X,U/Y, U/Z, X/Y}
	\draw (\from) -- (\to);
	% \draw [-,dashed] (C) to (R);
	% \draw [->,dashed] (C) to (X);
	% \draw [->] (Z) to [out=-30,in=-150] (Y);  
	% \draw [->] (Z) to [out=-30,in=-150] (X2);  
	\end{tikzpicture}
		\subcaption{A nondifferential  proxy  of the confounder}
\end{minipage}
\caption{Example causal diagrams for auxiliary variables. }\label{fig:dag}
\end{figure}
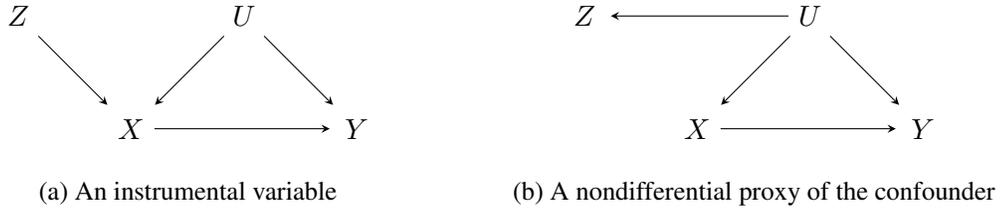

Equivalence  is a high-level assumption  stating that the treatment-confounder  distribution   lies in a model that is identified  upon a one-to-one transformation of $U$.
Because ignorability holds conditional on any one-to-one transformation of $U$, this  allows us to use an arbitrary admissible   treatment-confounder distribution to identify the treatment effects.
The equivalence property restricts the class of treatment-confounder distributions; 
as one example,  it is not met if the dimension of confounders exceeds that of the treatments.
Nonetheless,  the equivalence property  admits   a large  class of models.
In particular, it allows for any factor model or mixture model that is identified, where identification in the context of these models does not imply point identification but rather identification up to a rotation (factor models) or up to label switching (mixture models). 
Such model assumptions  are often used in bioinformatics applications where the unmeasured confounder represents population structure (GWAS) or lab batch effects \citep{wang2017confounder,wang2019jasa,luo2019batch}.
Identification results for factor and mixture models have been very well established \citep{anderson1956,kuroki2014measurement,titterington1985statistical,yakowitz1968identifiability}.
A major limitation of factor models is that they are in general not identified  when there are single-treatment confounders or when there are  causal relationships among the treatments  \citep{ogburn2019comment}. 
However, the equivalence assumption can accommodate models that allow for both of these features, 
for instance,  normal mixture models \citep{yakowitz1968identifiability};
see the supplement for an example.

Completeness is a fundamental concept   in statistics (see \cite{lehman1950completeness}),
and primitive conditions  are readily available in the literature, including  the fact that it holds for very general exponential families of distributions and for many regression models; see e.g.,  \citet{newey2003instrumental}  and \cite{d2011completeness}.
\cite{chen2014local} and \cite{andrews2017examples} have  shown that if $Z$ and $U$ are continuously distributed and the dimension of $Z$ is larger than that of $U$,  then  under    a mild regularity condition  the completeness condition holds generically in the sense that   the set of distributions for which completeness fails has a property analogous to having zero Lebesgue  measure.
By appealing to such results,   completeness holds in a large class of distributions and thus one may argue that it is  commonly satisfied.
The role of  completeness in this paper is  analogous to its   wide use in    a variety of nonparametric and semiparametric identification problems,  for instance, in   IV regression   \citep{newey2003instrumental},
IV quantile regression \citep{chernozhukov2005iv}, measurement error   problem \citep{hu2008instrumental}, missing data   \citep{miao2016varieties,d2010new}, and proximal inference \citep{miao2018proxy}. 
Our completeness  assumption  means that, conditional on $X$,  any variability in $U$ is captured by variability in $Z$,  
analogous to the relevance   condition in the  instrumental variable identification.
It  is easiest understood in the  categorical case.
For the binary  confounder case, completeness holds if   $U$ and $Z$ are correlated within each level of $X$.
When  both $U$ and $Z$ have $k$ levels,  
completeness means that  the matrix $[\pr(u_i\mid x,z_j)]_{k\times k}$ consisting of the conditional probabilities is invertible. 
This is stronger than dependence of $Z$ and $U$ given $X$.
Roughly speaking, dependence reveals that variability in $U$ is accompanied by  variability in $Z$, 
and  completeness reinforces that  any infinitesimal  variability in $U$ is accompanied by   variability in $Z$.
As a consequence, completeness in general fails if the number of levels or dimension of $Z$ is smaller  than that of $U$ or $Z$
is a coarsening of $U$.
In practice, completeness is more plausible if  practitioners     measure a rich set of potential auxiliary variables for the purpose of confounding adjustment. 
In the usual case that the dimension of $U$ is much smaller than that of $X$, the dimension of $Z$ can also be  small.
Completeness can be checked in specific models, for instance,  in the categorical case.
However, \cite{canay2013testability} show that  for unrestricted models the  completeness condition  is in fact untestable.
In the supplement, we further elaborate the discussion on completeness and provide both positive and negative examples to facilitate its   interpretation and use in practice.

Proposition 1 formalizes the equivalence and completeness conditions for the linear factor model that is widely used in GWAS and computational biology applications.

\begin{proposition}\label{prop:factor}
Consider   a factor model $X= \alpha U +\eta Z + \varepsilon$  for a vector of $p$ observed variables $X$,   
$q$ unobserved confounders  $U$, and  $r(\geq q)$  instrumental variables $Z$ such that $Z\ind U\ind \varepsilon$.
Without loss of generality we let $E(U)=0$, $\Sigma_U=I_q$,   $E(\varepsilon)=0$,   and $\Sigma_\varepsilon$ be diagonal.
Assuming there remain two disjoint submatrices of rank $q$  after deleting any row of $\alpha$,
we have that
\begin{itemize}
\item[(i)] $ \alpha \alpha^\T$ and $\Sigma_\varepsilon$ are uniquely determined from $\Sigma_{X-\eta Z}=\alpha\alpha^\T+\Sigma_\varepsilon$, 
and  any admissible  value for $\alpha$ can be written as $\tilde \alpha =\alpha R$ with $R$ an arbitrary $q\times q$ orthogonal matrix;
\item[(ii)] if  the components of $\varepsilon$ are mutually independent and  the joint characteristic function of  $X$ does not vanish, then any admissible joint distribution can be written as $\tilde \pr(x,u\mid z)=\pr(X=x, R^\T U=u\mid Z=z;\alpha)$ with $R$ an arbitrary $q\times q$ orthogonal matrix;
\item[(iii)] if $U,Z$, and $\varepsilon^\T$ are normal variables and $\eta^\T\gamma$ has full rank of $q$, then $\pr(u\mid x, z) \thicksim N(\gamma^\T x - \gamma^\T\eta z, \Sigma)$  and  is complete in $z$, where  $\gamma = (\Sigma_{X-\eta Z})^{-1} \alpha$, $\Sigma = I_q -  \alpha^\T (\Sigma_{X-\eta Z})^{-1}\alpha$.
\end{itemize}

\end{proposition}
The first result follows from    \cite{anderson1956} by noting that $\eta$ is identified by the regression of $X$ on $Z$,
the third result can be obtained from the  completeness property of exponential families  \citep[Theorem 2.2]{newey2003instrumental},  and we prove the second result in the supplement.
The first two results hold without $Z$, i.e., when $\eta=0$.
This proposition demonstrates that,  for the linear factor model, any admissible  value   for $\alpha$ must be  some rotation of the true value and any admissible treatment-confounder distribution  must be  the joint distribution of $X$ and some rotation of $U$.
Proposition \ref{prop:factor}  requires that $p\geq 2q+1$ and that each  confounder is correlated with at least three observed variables, and therefore  implies ``no single- or dual-treatment confounders.'' 
This is stronger than the ``no  single-treatment confounder'' assumption of \cite{wang2019jasa};
however,  \cite{grimmer2020ive} argue that in fact \cite{wang2019jasa} require the much stronger assumption of ``no finite-treatment confounding."

\subsection{Identification}
Leveraging  auxiliary variables gives  the following identification result.

\begin{theorem}\label{thm:auxil}
Under   Assumptions \ref{assump:ign} and \ref{assump:auxil}, for any admissible joint distribution $\tilde \pr(x,u\mid z)$  that  solves $\pr(x\mid z)=\int_u  \tpr(x,u\mid z) du$, 
there  exists a   unique solution $\tpr(y\mid u, x)$ to the equation
\begin{eqnarray}\label{eq:int}
\pr(y\mid x,z) = \int_u \tpr(y\mid u,x)\tilde \pr(u\mid x, z)du,
\end{eqnarray}
and the potential outcome distribution is identified by
\begin{eqnarray}\label{eq:idn}
\pr\{Y(x)=y\} = \int_u \tilde \pr(y\mid u, x)\tilde\pr(u)du.
\end{eqnarray}
where $\tpr(u)$ is obtained from $\tilde \pr(x,u\mid z)$ and $\pr(z)$.
\end{theorem}

Although the equivalence and completeness assumptions impose restrictions on the treatment-confounder distribution $\pr(x,u\mid z)$, the outcome model  $\pr(y\mid u,x)$    is left unrestricted in the sense that the parameter space of $\pr(y\mid u,x)$ is  all possible conditional densities of $Y$ given $X$ and a $q$-dimensional confounder $U$.
Theorem \ref{thm:auxil} depicts three steps of the auxiliary variables approach.
First we  obtain an arbitrary admissible  distribution $\tilde \pr(x,u\mid z)$;
then by solving equation \eqref{eq:int} we identify $\tilde\pr(y\mid u,x)$, which encodes  the treatment effect within each stratum of the confounder;  and finally we integrate the stratified  effect to obtain the treatment effect in the population. 
The auxiliary variables approach does not estimate the confounder, or even a surrogate confounder, and thus dispenses with the need for an infinite number of treatments and  avoids the forced positivity violations that were described by \cite{damour2019multi,damour2019comment}  and \cite{ogburn2020counterexamples}.

The auxiliary variable is indispensable in the second stage of the approach; without it  one  has to solve  
$\pr(y\mid x) = \int_u \tpr(y\mid u,x) \tilde\pr(u\mid x)du$ for the outcome model. 
The solution to this equation is  not unique given $\pr(y\mid x)$ and $\tilde\pr(u\mid x)$.
However,  by incorporating an auxiliary variable satisfying the exclusion restriction, we obtain equation \eqref{eq:int},
a Fredholm integral equation of the first kind \citep[chapter 15]{kress1989linear}. 
The solution of this equation is unique  under  the completeness condition and thus identifies  the outcome model, 
up to an invertible transformation of the confounder.
Equation \eqref{eq:int} also offers testable implications for Assumption \ref{assump:auxil}: 
if the equation does not have a solution, then  Assumption \ref{assump:auxil} must be partially violated.

Unlike  the g-formula \eqref{eq:effect},  we  do not identify the true outcome model $\pr(y\mid u,x)$ or the true confounder distribution $\pr(u)$, 
but instead we obtain  $\tilde\pr(y\mid u,x)=\pr\{y\mid V(U)=u, x\}$ and  $\tilde\pr(x,u\mid z)=\pr\{x,V(U)=u\mid z\}$ for some invertible transformation $V(U)$. Nonetheless,  for any such admissible pair of  outcome model and treatment-confounder distribution, we can still identify the potential outcome distribution, because 
ignorability holds conditional on any such  transformation of $U$.
The equivalence assumption guarantees that any admissible  distribution $\tpr(x,u\mid z)$ 
can be  used  for identifying the potential outcome distribution;
we do not need to use the truth $\pr(x,u\mid z)$ and thus bypass the challenge to identifying it.

Although Theorem \ref{thm:auxil} shows that  the potential outcome distribution is identified, 
the integral equation \eqref{eq:int} does not admit an analytic solution in general and one has to resort to numerical methods. For instance, \cite{chae2019algorithm} provide an estimation algorithm that is conjectured  to provide a consistent estimator of the unknown function under mild conditions.  
Nonetheless,  in certain special cases, a closed-form identification formula  can be derived.
\begin{example}\label{examp:iv}
Suppose  $p$ treatments, one confounder, one instrumental variable,  and one outcome are generated as $X= \alpha U + \eta Z + \varepsilon$ and $Y=m(X,U,e)$,  where $(\varepsilon^\T, U, Z)$  is a vector  of   independent normal variables with mean zero,  $\Sigma_U=1$, $m$ is  unknown,  and $e \ind (\varepsilon^\T, U, Z)$.
We require that at least three entries of  $\alpha$ are  nonzero   and that $\eta^\T\gamma\neq0$, in which case, 
the equivalence and completeness assumptions are met according to Proposition \ref{prop:factor}.
Given  an admissible value $\tilde \alpha$, we let  $\tilde \gamma= (\Sigma_{X-\eta Z})^{-1}\tilde \alpha$,  $\tilde \sigma^2 = 1 - \tilde \alpha ^\T (\Sigma_{X-\eta Z})^{-1}\tilde \alpha$,  then $\tpr(u\mid x,z) \thicksim N(\tilde \gamma^\T x -  \tilde \gamma^\T \eta  z, \tilde \sigma^2)$  is an admissible distribution for $\pr(u\mid x,z)$. 
Let
\begin{align*}
h_1(t)&=\int_{-\infty}^{+\infty}\exp(-{\rm i}tz)\phi(z)dz,\\
h_2(y,x, t)&= - \frac{\tilde \gamma^\T \eta}{\tilde \sigma}\int_{-\infty}^{+\infty} \exp\left\{-{\rm i}t \frac{\tilde \gamma^\T x -  \tilde \gamma^\T \eta z}{\tilde \sigma}\right\} \pr(y\mid x, z)dz,
\end{align*}
be the Fourier transforms of  the standard normal density   function $\phi$ and $\pr(y\mid x,z)$ respectively, 
where  ${\rm i}=(-1)^{1/2}$ denotes the imaginary unity.
Then the solution to \eqref{eq:int} with   $\tpr(u\mid x,z)$ given above is
\[\tpr(y\mid x,u)=\frac{1}{2\pi}\int_{-\infty}^{+\infty}\exp\left(\frac{\i t u}{\tilde \sigma}\right)\frac{h_2(y,x, t)}{h_1(t)} dt,\]
and the potential outcome distribution is 
\[\pr\{Y(x)=y\} = \int_{-\infty}^{+\infty} \tpr(y\mid u,x)\phi(u)du.\]
\end{example}

Detailed derivation for Example \ref{examp:iv} is deferred to  the supplement.
If  a linear outcome model is assumed and the structural causal parameter is of interest, 
then identification and estimation   are simplified as we will discuss in Section 5.2.
In the supplement, we   include an  additional example where identification  rests on  a  confounder proxy variable.

\subsection{A comparison to the conventional instrumental variable and proximal inference approaches}

We briefly describe  the difference between our auxiliary variables  approach and the  instrumental variable and negative control or proximal inference approaches.
In addition to  the  exclusion restriction $(Z\ind Y\mid (X,U))$, the  instrumental variable approach requires additional assumptions to achieve identification,  such as an additive outcome model not allowing for interaction of the treatment and the confounder $E(Y\mid u,x)= m(x) + u$ as well as completeness in $z$ of $\pr(x\mid z)$  \citep{newey2003instrumental}.
Alternative   strands of  using IV for confounding adjustment include  nonseparable outcome models \citep{imbens2009identification,chernozhukov2005iv,wang2018bounded}  and local average treatment effect models \citep{angrist1996identification,ogburn2015doubly}.
However, these two approaches    typically focus on a single (or binary) treatment and
we are not aware of   any extensions for    multiple treatments.
Therefore,   we   compare our approaches to the additive model, which has a straightforward extension to multiple treatments.
The completeness of $\pr(x\mid z)$ guarantees  uniqueness of the solution to  $E(Y\mid z)=\int_x m(x)\pr(x\mid z)dx$, 
an integral equation identifying  $m(x)$.
In contrast, our approach does not rest on  outcome model restrictions and hence  accommodates interactions.
Moreover,  the completeness  of $\pr(x\mid z)$  in $z$ entails at least as many  instrumental variables as  there are confounded treatments, 
and requires  each confounded treatment to be correlated with at least one instrumental variable.
%Such assumptions  may not be  practical  in the multi-treatment setting given the difficulty of identifying instrumental variables.
But for  our auxiliary variables approach, completeness of $\pr(u\mid x,z)$ in $z$ requires the dimension of $Z$ to be  as great as that of $U$, which can be much smaller than that of the treatments.
In the special case of a single  confounder, completeness of $\pr(u\mid x,z)$  can be  more plausibly achieved with only one auxiliary variable and with only one treatment  correlated with it.

Proximal inference    \citep{miao2018proxy,shi2020multiply,tchetgen2020introduction} allows for unrestricted outcome models, 
but entails at least two confounder proxies $(W, Z)$ with exclusion restrictions:
$W\ind (X,Z)\mid U$ and $Z\ind Y\mid (X,U)$,  even in the single confounder setting.
This approach additionally assumes existence of a function $h(w,y,x)$, called the confounding bridge function, such that  $\pr(y\mid u,x)=\int_w h(w,y,x)\pr(w\mid u)dw$, i.e.,  $h(w,y,x)$  suffices to depict the relationship between the confounding on $Y$ and $W$.
The integral equation $\pr(y\mid x,z)=\int_w h(w,y,x)\pr(w\mid x,z)dw$ is solved for the confounding bridge $h(w,y,x)$, and the potential outcome distribution is obtained by $\pr\{Y(x)=y\}=\int_w h(w,y,x)\pr(w)dw$, 
where completeness of $\pr(u\mid x,z)$ in $z$ is also required for identification of $\pr\{Y(x)\}$.

A strength of these two approaches is that they leave the treatment-confounder distribution unrestricted. 
But when the correlation structure of multiple treatments is informative about the presence and nature of confounding, as is generally the case in GWAS and computational biology applications,
our method can exploit this correlation structure to remove the confounding bias, 
while the conventional instrumental variable and  proximal inference approaches are agnostic to the treatment-confounder distribution and therefore unable to leverage any information it contains.

\section{Identification under the null treatments assumption}\label{sec:specificity}

\subsection{Identification}
Without auxiliary variables,
we let $\pr(x,u\mid \alpha)$ denote a model for the treatment-confounder distribution and $\pr(x;\alpha)$ the resulting marginal distribution indexed by a possibly infinite-dimensional parameter $\alpha$.
Let $\C=\{i:\pr(u\mid x) \text{ varies with } x_i\}$  denote the indices of confounded treatments  that are   associated with the confounder and $\A=\{i: \pr(y\mid u,x) \text{ varies with } x_i\}$  the active ones that affect the outcome.  A key feature of this identification strategy is that the analyst does not need to know which treatments are confounded or which are active.
We  make the following  assumptions.
\begin{assumption}\label{assump:null}
\begin{enumerate}
\item[(i)] Null treatments:   the cardinality of the intersection $\C\cap \A$  does not exceed $(|\C|-q)/2$, where $|\C|$ is the cardinality of  $\C$ and must be larger than the dimension of $U$;
\item[(ii)] Equivalence:  for any $\alpha$, any   $ \tilde \pr(x, u)$  that solves  $\pr(x;\alpha)=\int_u \tpr(x,u) du$  can be written as $ \tilde\pr(x, u) =  \pr\{X=x, V(U)=u;\alpha\}$ for some  invertible but not necessarily known  function $V$;
\item [(iii)] Completeness: for any $\alpha$,  $\pr(u\mid x;\alpha)$ is complete in   any $q$-dimensional subvector $x_\S$ of the confounded treatments $x_\C$.
\end{enumerate}
\end{assumption}

Analogous to Assumption \ref{assump:auxil},   the equivalence and completeness assumptions restrict the treatment-confounder distribution, but hold for certain well-known classes of models like   factor or mixture models.
Under  the equivalence assumption, one can   identify the confounded treatments set $\C$ by using an arbitrary admissible joint distribution $\tpr(x,u)$ without knowing the truth.
The null treatments assumption entails that  fewer than  half  of the  confounded treatments   can have  causal effects on the outcome but does not require  knowledge of  which  treatments are active.
The assumption  is reasonable  in many empirical studies where only a few but not many of the treatments can causally affect the outcome.
For example, in   GWAS or analyzing electronic health record  databases for off-label drugs that improve COVID-19 outcomes, one may expect that most treatments are null without knowing which are active treatments.
A counterpart of the null treatments assumption was previously considered by \cite{wang2017confounder} in the context of effects of a treatment on multiple outcomes.
They consider a linear factor model  $Y=\alpha U +\beta X +\varepsilon$  for   $p$ outcomes $(Y)$ and
show that $\beta$ is identified  if only a small proportion of its elements are nonzero.
By  reversing  the roles of treatments and outcome, their approach  can be adapted to  test but not identify multi-treatment  effects,   because   the   coefficient in the regression of treatments on the outcome and the confounder  is not the causal effect of interest.
Another related concept is the $s$-sparsity \citep{kang2016instrumental} used in Mendelian randomization, which assumes that at most $s$  single nucleotide polymorphisms can directly affect the outcome of interest.  
In contrast, the null treatments assumption does not necessarily require the effects to be sparse and imposes no restrictions on the unconfounded treatments.

Leveraging  the null treatments assumption gives  the following identification result.

\begin{theorem}\label{thm:null}
Under Assumptions \ref{assump:ign} and \ref{assump:null}, 
for any joint distribution $\tilde \pr(x,u)$  that  solves $\pr(x)=\int_u \tpr(x,u) du$, 
there  exists a   unique solution $\tilde \pr(y\mid u, x)$ to the equation
\begin{eqnarray}\label{eq:int2}
\pr(y\mid x) = \int_u \tpr(y\mid u, x)\tilde\pr(u\mid x)du,
\end{eqnarray}
and the potential outcome distribution is identified by
\begin{eqnarray}\label{eq:idn2}
\pr\{Y(x)=y\} = \int_u \tilde\pr(y\mid u, x)\tilde\pr(u)du.
\end{eqnarray}
\end{theorem}

This theorem states that treatment effects are identified if fewer than half of the confounded treatments can affect the outcome.
The  outcome model  is left unrestricted except for the restriction imposed by the null treatments assumption.

Analogous to the auxiliary variables approaches, the equivalence assumption allows us  to use an arbitrary admissible treatment-confounder distribution for identification, 
the null treatments assumption allows us to construct  Fredholm integral equations of the first kind to solve 
for the outcome model, and the completeness assumption guarantees uniqueness of the  solution.
Without the null treatments assumption,  the solution to \eqref{eq:int2} is not unique as we noted above.

Theorem \ref{thm:null}  holds if we replace $q$ with an integer $s\geq q$ in Assumption \ref{assump:null}, 
in which case,  completeness   is weakened but the null treatments assumption is strengthened.
If the average treatment effect is of interest,  we could define the active treatments set as $\A=\{i: E(Y\mid u,x) \text{ varies with } x_i\}$ and analogously   identify $E\{Y(x)\}$.
If a  linear outcome model   is assumed and the structural parameter is of interest, the identification and estimation are simplified. We demonstrate this  in Section \ref{sec:linear}.

To illustrate, the following is a simple example where Assumption \ref{assump:null} holds and  identification is achieved. 
We are not aware of any previous identification results for this setting, in particular  when  no parametric assumption is imposed on the outcome model.
\begin{example}\label{exam:specificity}
Suppose  $p$  treatments, one confounder, and one outcome are generated as $X= \alpha U +  \varepsilon$ and $Y=m(X,U,e)$,  
where $(\varepsilon^\T, U)$ is  a vector  of independent  normal variables with mean zero, $\Sigma_U=1$, $m$ is  unknown,  and $e \ind (\varepsilon, U)$. 
Suppose the  entries of $\alpha$ are nonzero and $m$  can depend  on at most $(p-1)/2$  treatments, 
but we do not know which ones.
In this setting, $\pr(u\mid x)\thicksim N(\gamma^\T x, \sigma^2)$ with $\gamma=\Sigma_X^{-1} \alpha$ and $\sigma^2= 1- \alpha^\T \Sigma_X^{-1} \alpha$.
Because the entries of $\alpha$ are nonzero, the equivalence assumption is met (Proposition \ref{prop:factor} with $\eta=0$) and  all  entries of $\gamma$ must be nonzero (lemma 2 in the supplement);  thus, $\pr(u\mid x)$ is complete   in each treatment.
As a result, the potential outcome distributions and  treatment effects are identified.
\end{example}

The null treatments approach proceeds  by first obtaining an admissible $\tilde\pr(x,u)$,  
then finding the solution to  \eqref{eq:int2}, and finally calculating  $\pr\{Y(x)\}$ from \eqref{eq:idn2}.
Equation \eqref{eq:int2} cannot be solved directly as we do not know which treatments are null.
As an informal, heuristic intuition for how this identification strategy works in this example, note that, if we knew the identity of any one of the null treatments we could treat it as an auxiliary variable and apply the auxiliary variables approach. In the absence of such knowledge we can imagine using each $X$ as an auxiliary variable; under the null treatments assumption more than half of the $X$'s must give the same, correct solution to \eqref{eq:int2}. We describe a constructive method that formalizes this idea.

Let $x_\S$ denote an arbitrary $q$-dimensional subvector of $x_\C$ and $x_{\bar \S}$  the rest components of $x$ except for $x_\S$. Given $\tpr(x,u)$ and $\pr(y\mid x)$,  we solve 
\begin{eqnarray}\label{eq:int3}
\pr(y\mid x) = \int_u \tpr(y\mid u, x_{\bar \S})\tilde\pr(u\mid x_\S, x_{\bar \S})du,
\end{eqnarray}
for $ \tpr(y\mid u, x_{\bar \S})$ for all choices of  $x_\S$.

\begin{proposition}\label{prop:null}
Under Assumptions \ref{assump:ign} and \ref{assump:null},  for any choice for $x_\S$, if the solution to \eqref{eq:int3}  exists and depends on at most $(|\C|-q)/2$ ones of the confounded treatments,
then it must solve \eqref{eq:int2}. 
\end{proposition}

\subsection{Hypothesis testing without auxiliary variables and null treatments assumptions}
An immediate extension of the null treatments approach provides  a test of the sharp null hypothesis  of no joint effects, which requires neither auxiliary variables nor the null treatments assumption. 
The sharp null hypothesis is $\mathbb H_0: \pr(y\mid u,x)=\pr(y\mid u)$ for all  $x$.

\begin{proposition}\label{thm:hyp}
	Under Assumption \ref{assump:ign} and  (ii)--(iii) of Assumption \ref{assump:null}
	and given an admissible joint distribution $\tilde \pr(x,u)$, 
	if  the null hypothesis $\mathbb H_0$ is correct, 
	then for any $q$-dimensional subvector $x_\S$ of the confounded treatments $x_\C$, 
	the solution to  the following equation	exists and is  unique,  
	\begin{eqnarray}\label{eq:hyp1}
	\pr(y\mid x) = \int_u \tpr(y\mid u, x_{\bar \S})\tilde\pr(u\mid x_{ \S},x_{\bar \S})du,
	\end{eqnarray}
and  the solution must satisfy the  following equality,
	\begin{eqnarray}\label{eq:hyp2}
	\pr(y)=  \int_u \tilde\pr(y\mid u, x_{\bar \S})\tilde\pr(u)du.
	\end{eqnarray}
\end{proposition}
This result allows us to construct valid hypothesis tests even in the absence of auxiliary variables or a commitment to the null treatments assumption. 
Under Assumption \ref{assump:ign} and (ii)--(iii) of   Assumption \ref{assump:null}, 
evidence  against the existence of the solution to \eqref{eq:hyp1}  or against the   equality \eqref{eq:hyp2}  is evidence against $\mathbb H_0$.
The proof is immediate by noting that, under the null hypothesis $\mathbb H_0$,  the null treatments assumption is trivially satisfied. 
Thus,  if $\mathbb H_0$ is correct, the solution to \eqref{eq:hyp1} does not depend $x$, and the right hand side of \eqref{eq:hyp2}  identifying the potential outcome distribution $\pr\{Y(x)=y\}$ must be equal to the observed outcome distribution $\pr(Y=y)$.

\section{Estimation}\label{sec:estimation}
\subsection{General estimation strategies}
In this section we adhere to a common principle of causal inference, and indeed statistics more broadly, that is nicely summed up by  \cite{cox2011principles} (via \citealp{imai2019discussion}):
\begin{quote}
If an issue can be addressed nonparametrically then it will often be better to tackle it parametrically; however, if it cannot be resolved nonparametrically then it is usually dangerous to resolve it parametrically. (p. 96)
\end{quote}
Having proven identification without recourse to parametric outcome models, we will see that estimation does, in fact, require them. This is similar to myriad other causal inference problems, where nonparametric identification is commonly followed by estimation via parametric or sometimes semiparametric models. One key difference is that, in the absence of parametric assumptions, there is no closed-form identifying expression for treatment effects in this setting, because in general there is no closed-form solution to the integral equations that are involved in identification. We first describe an estimation procedure in full generality and then introduce some choices of parametric outcome models that make it feasible in practice.

Theorem \ref{thm:auxil} points to the following   auxiliary variables algorithm for estimation of $\pr\{Y(x)\}$.

\begin{framed}
The auxiliary variables algorithm:
\begin{itemize}[leftmargin=6em]
	  \item[\emph{Aux-1}]  Obtain an arbitrary admissible joint distribution $\tpr(x,u, z)$.
  
	\item[\emph{Aux-2}] Use the estimate from Step 1, along with an estimate of $\pr(y\mid x,z)$, to solve Equation \eqref{eq:int} for an estimate of $\tpr(y \mid u, x)$.
	
	\item[\emph{Aux-3}] Plug the estimate of $ \tpr(y \mid u, x)$ from Step 2 and the estimate of $\tilde \pr(u)$ derived from $\tpr(u, x,z)$  into Equation \eqref{eq:idn} to estimate $\pr\{Y(x)\}$.
\end{itemize}
\end{framed}

Theorem \ref{thm:null} points to the null treatments  algorithm with a similar set of steps for estimation of $f\{Y(x)\}$.

\begin{framed}
The null treatments  algorithm:
\begin{itemize}[leftmargin=6em]

	\item[\emph{Null-1}] Obtain an arbitrary admissible joint distribution $\tpr(x,u)$.
	
	\item[\emph{Null-2}] Use the estimate of $ \tpr(u \mid x)$ from Step 1, along with an estimate of $\pr(y\mid x)$, to solve Equation \eqref{eq:int2} for an estimate of $ \tpr(y \mid u, x)$. The constructive method described in Proposition \ref{prop:null} can be implemented to solve \eqref{eq:int2}.
	
	\item[\emph{Null-3}] Plug the estimate of $\tilde \pr(u)$ from Step 1 and $\tilde \pr(y \mid u, x)$ from Step 2 into Equation \eqref{eq:idn2} to estimate $\pr\{Y(x)\}$.
\end{itemize}
\end{framed}

Steps \emph{Aux-1} and \emph{Null-1} require only the equivalence assumption, which places nontrivial restrictions on the treatment-confounder distribution. 
To estimate $\tpr(x,u)$, one needs to correctly specify 
a treatment-confounder model that meets the equivalence assumption, 
such as  a factor or  mixture model.
Under standard   factor or mixture models, estimation of  $\tpr(x,u)$ is well established, 
and we refer to the large existing body of literature for estimation techniques and properties \citep{anderson1956, kim1978factor,titterington1985statistical}. 
Note that, crucially, Step 1 estimates the distribution of $U$ (joint with $X$), but does not estimate $U$ itself. This is advantageous   as it does not require infinite number of  treatments and engender the resulting positivity problems.

Steps \emph{Aux-2} and \emph{Null-2} involve, first, estimating $\pr(y\mid x)$ or $\pr(y\mid x,z)$, 
which can be done parametrically or nonparametrically using standard density estimation techniques. 
More challenging is the second step, solving integral equations   \eqref{eq:int} and \eqref{eq:int2},  
which do not admit analytic solutions in general.  
These equations are of the form of Fredholm integral equations of the first kind \citep[chapter 15]{kress1989linear}.
This kind of equation is known to be ill-posed due to noncontinuity of the solution and difficulty of computation.
In the contexts of nonparametric instrumental variable regression, regularization methods have been established to solve the equation and we refer to  \cite{newey2003instrumental,carrasco2007}  for a broad view of this problem.
Numerical solution to such equations  is an active area of mathematical and statistical research and is largely beyond the scope of this paper.  
However, we note that \cite{chae2019algorithm} provide \texttt{R} code for a numerical method that is 
conjectured to  provide a consistent estimator of the unknown function under mild conditions.
In the next subsections, we describe modeling assumptions under which the integral equation can be avoided altogether.

Steps \emph{Aux-3} and \emph{Null-3} are essentially applications of the g-formula; estimation of this integral is standard in causal inference problems.

Below we provide two examples of how parametric models--in this case linear--can obviate the need to solve integral equations,  permit estimation using standard software, and  admit   consistent estimators.

\subsection{The auxiliary variables approach with linear  models} \label{sec:aux_linear}
Consider the following model for a $p$-dimensional treatment $X$, a $q$-dimensional confounder $U$, and an $r$-dimensional instrumental variable $Z$, with $p\geq 2q+1$ and $r\geq q$:
\begin{gather}
X= \alpha U + \eta Z + \varepsilon,\quad \Sigma_U=I_q, \quad E(U)=0, \quad U\ind Z\ind \varepsilon,\quad \Sigma_\varepsilon\text{ diagonal,} \label{eq:factor}\\
\text{there remain two disjoint submatrices of rank $q$  after deleting any row of $\alpha$},\label{eq:alpha}\\
\text{$\eta^\T\gamma$ has full rank of $q$, where $\gamma=(\Sigma_{X-\eta Z})^{-1}\alpha = (\Sigma_\varepsilon  + \alpha\alpha^\T)^{-1}\alpha$},\label{eq:iv}\\
E(Y\mid U,X,Z) = \beta^\T X + \delta^\T U. \label{eq:outcome}
\end{gather}
This is the IV setting from Example \ref{examp:iv}.  
Intercepts are not included in the models as one can center  $(X,Y,Z)$ to have mean zero.
If $(\varepsilon^\T, U, Z)$ are normally distributed, then   \eqref{eq:factor}--\eqref{eq:alpha}  imply   equivalence  and \eqref{eq:iv} implies   completeness in Theorem \ref{thm:auxil}, and  as a special case, identification and estimation of $\pr\{Y(x)\}$ follows from  Theorem \ref{thm:auxil} and the auxiliary variables algorithm, respectively.
However, the estimation procedure below works even when     the error distributions are left  unspecified,   in which  case \eqref{eq:factor}--\eqref{eq:iv}  do not suffice for identification of $\pr\{Y(x)\}$.
We additionally  assume  the linear outcome model \eqref{eq:outcome} and focus on  the structural parameter $\beta$ but not the potential outcome distribution $\pr\{Y(x)\}$.
Then  \eqref{eq:factor}--\eqref{eq:iv},  viewed as a parallel version of Assumption \ref{assump:auxil}, 
guarantee identification of $\beta$.
Condition \eqref{eq:iv}  in principle can be tested after obtaining an estimator of $\eta$ and $\gamma$.
Note that, $r$ may be smaller than $p$, in which case the conventional IV estimator $\hat \beta_{\rm iv} =(Z^\T X)^{-1}Z^\T Y$ does not work.

Estimation of $\beta$ is parallel to the auxiliary variables algorithm.
We first obtain  $\hat \eta$ by regression of $X$ on $Z$ and  obtain   $\hat \gamma$    by   factor analysis of the residuals $X- \hat\eta Z$.
There must exist some orthogonal matrix $R$ so that  $\hat  \gamma$ converges to $\gamma R$. 
This corresponds to   \emph{Aux-1}.
Let $(\xi^X,   \xi^Z)$  denote the    coefficients by   regression  of  $Y$ on $(X,Z)$ and $(\hat \xi^X,   \hat \xi^Z)$ be the corresponding estimator.
Note that $\xi^X =    \beta +  \gamma   \delta$ and $\xi^Z = - \eta^\T  \gamma  \delta$, therefore we solve 
\begin{eqnarray}\label{eq:lineariv}
\hat \xi^X =  \hat \beta + \hat\gamma \hat \delta,\quad \hat \xi^Z = -\hat\eta^\T \hat\gamma \hat\delta
\end{eqnarray}
for $(\hat \beta, \hat\delta)$.
This corresponds to   \emph{Aux-2}: estimation of $\pr(y\mid x,z)$ is replaced by a linear regression of $Y$ on $(X,Z)$ and  solving the integral equation  for   $\tpr(y\mid u,x)$ is replaced by solving   linear equations for   finite-dimensional parameters $(\beta, \delta)$. 
%One can verify that   the confounding bias of $\tilde \beta^X$ is recovered with  any admissible value $\tilde \alpha$: $-\gamma \delta=\tilde \gamma(\tilde \gamma^\T \eta \eta^\T \tilde \gamma)^{-1} \tilde \gamma^\T \eta \tilde \beta^Z$ for $\tilde \gamma =(\Sigma_{X-\eta Z})^{-1}\tilde \alpha$.
We finally obtain
\[\hat\beta = \hat \xi^X  + \hat \gamma(\hat \gamma^\T \hat\eta\hat \eta^\T \hat \gamma)^{-1} \hat \gamma^\T\hat \eta \hat \xi^Z.\]
In the special case that the dimension of the instrumental variable $Z$ equals that of  the confounder $U$,
we obtain $\hat\beta = \hat \xi^X + \hat\gamma ( \hat\eta^\T \hat\gamma)^{-1}\hat \xi^Z$. 
Consistency and asymptotic normality follows from that of $(\hat\xi^X,\hat\xi^Z,\hat\gamma,\hat\eta)$.
Routine \texttt{R} software such as    \texttt{factanal} and  \texttt{lm} can be implemented for factor analysis and linear regression, respectively, and the variance of estimators can be bootstrapped.

\subsection{The null treatments approach with  linear  models} \label{sec:linear}

Consider the linear models:
\begin{gather}\label{mdl:linear}
\begin{split}
X = \alpha U + \varepsilon, \quad E(Y\mid X,U)= \beta^\T X  + \delta^\T  U,\\
\Sigma_U=I_q, \quad E(U)=0, \quad U\ind \varepsilon, \quad \Sigma_\varepsilon\text{ diagonal,} 
\end{split}
\end{gather}
where $U$ and $\varepsilon$ are     not necessarily normally distributed.
The coefficient $\beta$ encoding  the  average treatment effects is of interest. 
We let $\gamma=\Sigma_X^{-1}\alpha$, which denotes the coefficients of regressing $U$ on $X$.
Hereafter, we use $A_i$ to denote the $i$th row of  a matrix or a column vector $A$. 
Let $\C=\{i: \alpha_i\text{ is not a zero vector}\}$  denote  the indices of confounded treatments, 
$|\C|$   the cardinality  of $\C$, and  $\gamma_\C$  the  submatrix consisting of the corresponding rows of  $\gamma$.
We have the following  result.

\begin{theorem}\label{thm:ln}
The parameter $\beta$ is identified under model \eqref{mdl:linear} and the following assumptions:
\begin{enumerate}
\item[(i)]   at most  $(|\C|-q)/2$ entries of $\beta_\C$  are nonzero;
\item[(ii)] after deleting any row, there remain two disjoint submatrices of $\alpha$ of full rank;
\item[(iii)]  any  submatrix of $\gamma_\C$ consisting of $q$ rows has full rank.
\end{enumerate}
\end{theorem}
The identification result is not compromised  if   $q$ is replaced  with $s\geq q$ in the conditions.
If  the errors are normally distributed, then  conditions (i)--(iii)  of  Theorem \ref{thm:ln} imply the null treatments, the equivalence  and the completeness conditions in Assumption \ref{assump:null}, respectively, and  identification and estimation of $\pr\{Y(x)\}$ follows from  Theorem  \ref{thm:null} and the null treatments algorithm, respectively.
The estimation procedure below works for  unspecified  error distributions,  in which case conditions (i)--(iii) in Theorem \ref{thm:ln}  no longer suffice for identification of $\pr\{Y(x)\}$,  but as a parallel version of Assumption \ref{assump:null} they guarantee identification of $\beta$.
Condition (iii) of  Theorem \ref{thm:ln}   in principle can be tested after obtaining an estimator of   $\gamma$.

Estimation of $\beta$ is parallel to the null treatments algorithm.
Let $\xi$ denote the    coefficients by   regression  of  $Y$ on $X$.
Given $n$ independent and identically distributed samples, 
we first obtain $\hat\gamma$  by  factor analysis of $X$; this corresponds to \emph{Null-1}.
Then we obtain   $\hat \xi$ by  regression  of  $Y$ on $X$, which can be viewed as a crude estimator of $\beta$ with    asymptotic  bias  $\gamma\delta$. 
Note that  $  \xi_\C = \beta_\C +   \gamma_\C   \delta$ for confounded treatments, 
then   under  the null treatments   assumption (i) of Theorem \ref{thm:ln},  
estimation of $\hat \delta$  can be cast as a standard robust linear regression \citep{rousseeuw2005robust} given  $\hat \xi$ and $\hat\gamma$:
  $\hat \gamma_\C $ is  the design matrix and $\hat \xi_\C$ is  observations with outliers corresponding to nonzero entries of $\beta_\C$. 
This corresponds to  \emph{Null-2}, where the complexity of solving integral equations is resolved by robust linear regression.
Specifically,  given a   $n^{1/2}$-consistent estimator   $(\hat\xi,\hat\gamma)$,  we solve
\begin{eqnarray}
\hat \delta^{\rm lms} =\arg\min_{\delta}  \text{median}\  \{ (\hat \xi_i  - \hat \gamma_i \delta)^2,\  i\in \hat \C\},\quad  \hat \C=\{i: ||\hat \gamma_i||_2^2 > \log(n)/n\},
\end{eqnarray}
which is a least median of squares estimator   minimizing   median of the squared errors $(\hat \xi_i  - \hat \gamma_i \delta)^2$ among the confounded treatments consistently selected by $\hat\C$. 
For the robust linear regression,  the least quantile of squares, least trimmed squares, and the 
S-estimator can also be used to solve for $\delta$, which we refer to \citet[chapter 3.4]{rousseeuw2005robust}.
The corresponding estimate of $\beta$ is $\hat \beta^{\rm lms}=\hat\xi - \hat\gamma \hat\delta^{\rm lms}$.
In the supplement, we show consistency of $(\hat \delta^{\rm lms},\hat\beta^{\rm lms})$ under the assumptions of Theorem \ref{thm:ln} and an additional regularity condition given $n^{1/2}$-consistency of  $(\hat\xi,\hat\gamma)$.
However, they are not necessarily  asymptotically normal as we observe in numerical experiments, even if $(\hat\xi,\hat\gamma)$ are.
Therefore, to promote asymptotic normality we use  $\hat \beta^{\rm lms}$ as an initial value to select the null treatments and to  update estimates of $\delta$ and $\beta$ as follows:  
obtain $\hat\delta$ by solving $\hat\xi_i=\hat\gamma_i \hat\delta$ (by OLS)   that corresponds to   the smallest $(|\hat \C|+q)/2$ entries of $|\hat\beta_{\rm lms}|$ and obtain the ultimate estimator $\hat\beta=\hat\xi-\hat\gamma\hat\delta$. 
Routine \texttt{R} software \texttt{factanal} and  \texttt{lqs} can be implemented  for factor analysis and robust linear regression, respectively,  and the variance of the ultimate estimator can be bootstrapped.

\section{Simulations}\label{sec:sims}
\subsection{The auxiliary variables setting}
We  evaluate  performance of the proposed methods via simulations.
For  the auxiliary variables setting,  2   confounders $(U)$, 6  instrumental variables $(Z)$, 6 treatments $(X)$, an outcome $(Y)$,
and 2 outcome-inducing confounder proxies $(W)$ are generated   as follows,
\begin{eqnarray}
X = \alpha U + \eta Z + \varepsilon_X,\quad  Y = \beta^\T X +  \delta_Y U + \varepsilon_Y,\quad W = \delta_W U +\varepsilon_W,\\
U\thicksim N(0,I_2),  \quad Z\thicksim N(0,I_6),\quad  \varepsilon_X\thicksim N(0, I_6),  \quad \varepsilon_Y\thicksim N(0,1);\quad \varepsilon_W\thicksim N(0,I_2);
\end{eqnarray}
\[\begin{array}{l}\alpha = \begin{pmatrix}
0&0\\
1&0\\
1.5&1\\
2&-2\\
2.5&1\\
2&-1
\end{pmatrix},  \quad 
\eta = \begin{pmatrix}
1&0&0&0&0&0\\
0&1&0&0&0&1\\
0&0&1&0&0&0\\
0&0&0&1&0&0\\
0&0&0&0&1&0\\
1&0&0&0&0&1\\
\end{pmatrix}, \quad
\beta =\begin{pmatrix}
1\\
1\\
1\\
1\\
1\\
1\\
\end{pmatrix},
\quad \delta_Y =(1,1),\quad \delta_W = \begin{pmatrix}
2&0\\
0&2\\
\end{pmatrix}.
\end{array}\]
Under this setting, $(X_2, \ldots, X_6)$ are confounded but $X_1$ is not.
For estimation, we consider eight methods:
\begin{table}[H]
\centering
\begin{tabular}{ll}
IV$_1$ &  conventional IV  approach using all six IVs;\\
IV$_2$ &  conventional IV  approach using five IVs $(Z_1,\ldots, Z_5)$ and treating $(X_6,Z_6)$ as   covariates;\\
Aux$_1$ &the proposed auxiliary variables approach assuming two factors and using all six IVs;\\
Aux$_2$ & the  auxiliary variables approach assuming two factors,\\
& using  $(Z_5,Z_6)$ as IVs  and $(Z_1,\ldots, Z_4)$ as covariates;\\
Aux$_3$ &   the auxiliary variables approach with one factor, using  $Z_6$ as an IV and $(Z_1,\ldots, Z_5)$ as covariates;\\
PI$_1$ &     proximal inference  using  $(Z_5,Z_6)$ and $(W_1,W_2)$ as the treatment- and outcome-inducing  \\
& confounder proxies,   respectively and  $(Z_1,\ldots, Z_4)$ as covariates;\\
PI$_2$ &    proximal inference using  $Z_6$ and $W_1$ as the treatment- and outcome-inducing  confounder proxy, \\
& respectively and  $(Z_1,\ldots, Z_5)$ as covariates;\\
OLS&   ordinary least squares estimation by regression $Y$ on $X$ and $Z$.
\end{tabular}
\end{table}
Two stage least squares are used in the above IV or PI methods.
Because  \cite{grimmer2020ive} have shown that the deconfounder \citep{wang2019jasa} is asymptotically
equivalent  to  and  does not outperform   OLS, we   do not include a separate comparison with the deconfounder.
We replicate 1000 simulations at sample size 1000 and 2000. 
Figure \ref{fig:bias1} summarizes bias of the estimators of each parameter. 
As expected,    IV$_1$,  Aux$_1$, Aux$_2$, and PI$_1$  perform well with little bias for estimation of all parameters, because   sufficient number of IVs or confounder proxies are used in these four methods  and  the number of factors  are correctly specified in Aux$_1$ and Aux$_2$.
Note that,   Aux$_2$ uses only two IVs  while IV$_1$ uses all six IVs and   PI$_1$ uses   two additional confounder proxies.
We also compute the bootstrap confidence interval and evaluate the coverage probability for Aux$_1$ and Aux$_2$, summarized in  Table \ref{tbl:cvp1}.
The 95\% bootstrap confidence interval has coverage probabilities close to the nominal level of 0.95 under a moderate sample size.
When using the same set of IVs, we might expect Aux$_1$  to be 
more efficient than IV$_1$ as the former additionally incorporates the internal dependence structure of the treatments. \footnote{This is conjectured by an anonymous reviewer.}
However,  in our simulations when the sample size is increased to 10000, we observe that  Aux$_1$ has   a larger  mean squared error    than IV$_1$ for estimation of $(\beta_3, \beta_5)$ and Aux$_1$ has  a  larger  mean squared error  than Aux$_2$ for estimation of $\beta_5$.
It can also be seen from the boxplots of the bias in Figure \ref{fig:bias1}.
This may be because the proposed auxiliary variables estimation is not    efficient.
Therefore, in the future  it is of great interest to theoretically establish the  semiparametric efficiency bound and the efficient estimator under the auxiliary variables assumption and compare   to other competing methods.

In contrast, IV$_2$ fails to consistently estimate $\beta_6$ because the corresponding IV ($Z_6$)
is not correctly used but treated as a covariate, but surprisingly,  IV$_2$  is also biased for estimation of $\beta_1$ that is not confounded and can be estimated very well by all the other methods. 
Likewise,   PI$_2$  is biased for estimation of $(\beta_2,\ldots, \beta_5)$ due to insufficient number of confounder proxies.
Nonetheless, PI$_2$ has little bias for estimation of $\beta_6$.
This is because $Z_6$ used in PI$_2$ is a valid IV for $X_6$ and  the PI method is consistent even if the confounder proxies are inadequate, a property previously shown by \cite{miao2018negative}.
Except for $\beta_1$, Aux$_3$ is biased for estimation of the confounded parameters because the number of factors and IVs are incorrectly specified.
As expected, OLS is biased for estimation of all parameters except  for $\beta_1$.

We also evaluate performance of the estimators when  the IVs have a direct effect on the outcome.
In this case, the outcome model is changed to $Y = \beta^\T X + \lambda Z + \delta_Y U + \varepsilon_Y$ with 
$\lambda=(0.2,\ldots,0.2, 0.3)$,
while the other settings of the data generating process remain the same.
Figure S.1  in the supplement summarizes bias of the estimators. 
All estimators are biased for estimation of the confounded treatment effects, 
and no estimator outperforms the others in terms of bias.
However, the IV estimators are also biased for estimation of $\beta_1$ while the other estimators have little bias.

In summary,  all of   conventional IV,   proximal inference, and the proposed auxiliary variables approaches 
rely on the exclusion restriction assumption. 
If the exclusion restriction and the required conditions are satisfied, each of these approaches can properly address the multi-treatment confounding.
The conventional IV entails at least as many IVs  as   treatments, otherwise it can be biased even for unconfounded treatments;   proximal inference needs  at least as many treatment- and outcome-inducing  proxies as   confounders; the proposed auxiliary variables approach does not rest on outcome-inducing confounder proxies but  rests on a factor model and correct specification of  the factor number. 
Therefore, to obtain a reliable causal conclusion  in practice, we recommend   implementing  and comparing multiple applicable methods.

\begin{table}[H]
	\centering
	\caption{Coverage probability of the $95\%$ bootstrap confidence interval based on Aux$_1$ and Aux$_2$} \label{tbl:cvp1}
\begin{minipage}{0.6\textwidth}
	\begin{tabular}{cccccccc}
&&$\beta_1$&$\beta_2$&$\beta_3$&$\beta_4$&$\beta_5$&$\beta_6$\\
\multirow{2}{*}{Aux$_1$}&&0.948 &0.940& 0.940& 0.970& 0.933 &0.958\\
&&0.943 &0.949& 0.954& 0.950& 0.933 &0.949\\
\multirow{2}{*}{Aux$_2$}&&0.946 &0.966& 0.970& 0.963& 0.972 &0.960\\
&&0.942 &0.959& 0.958& 0.953& 0.950 &0.942
	\end{tabular}
\end{minipage}\hfil
\begin{minipage}{0.4\textwidth}
Note:  For each estimator, the first row is for sample size 1000  the second for  2000.
\end{minipage}
\end{table}

\graphicspath{{Results/simuauxsetting1/}}

\begin{figure}[H]
	\centering
\includegraphics[scale=0.3]{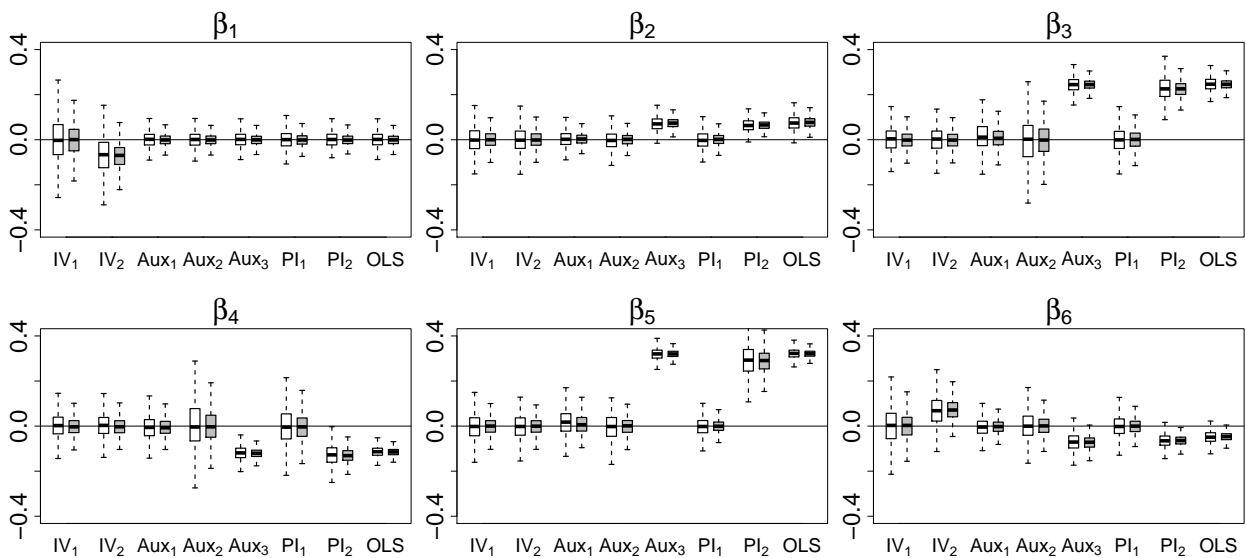}

\caption{Bias of estimators when the exclusion restriction holds. 
White boxes are for sample size 1000 and gray ones  for 2000.} \label{fig:bias1}
\end{figure}

\subsection{The null treatments setting}
We generate   2   confounders $(U)$, 8 treatments $(X)$,  and an outcome $(Y)$ as follows,
\[\begin{array}{c}
X = \alpha U   + \varepsilon_X,\quad  Y = \beta^\T X +  \delta_Y U + \varepsilon_Y,\quad 
U\thicksim N(0,I_2),   \quad  \varepsilon_X\thicksim N(0, I_6),  \quad \varepsilon_Y\thicksim N(0,1);\\
\alpha^\T = \begin{pmatrix}
0  &0.4 & 0.8 & 1.2 & 1.5& -0.4& -0.8& -1.2\\
0  &0.2 & 0.4 & 0.6 & 0.8& -0.5& -1.0& -1.2
\end{pmatrix},  \quad 
\delta_Y =(1,1).
\end{array}\]
We consider two choices for $\beta$,
\begin{itemize}
\item[] Case 1: $\beta_1=\beta_2=\beta_3=1$, $\beta_4=\ldots=\beta_8=0$;
\item[] Case 2: $\beta_1=\beta_2=\beta_3=1$, $\beta_4=\beta_5=0.2$, $\beta_6=\beta_7=\beta_8=0$.
\end{itemize}
Under these settings, $(X_2,\ldots, X_8)$ are confounded but $X_1$ is not;
the null treatments assumption is satisfied in Case 1 as only two of the confounded treatments are active,
but it is  violated in Case 2 where $\beta_4$ and $\beta_5$ also have a small effect on $Y$.

For estimation,  three methods are used: the null treatments   estimation with correct dimension of the confounder   (Null$_1$), the null treatments   estimation with one confounder (Null$_2$),  and OLS. 
Because no auxiliary variables are generated,  we do not compare to   IV,  the auxiliary variables, or proximal inference approaches.
We replicate 1000 simulations at sample size 2000 and 5000.
Figures \ref{fig:bias3} and Figure S.2 (in the supplement) summarize  bias of the estimators for Case 1 and Case 2, respectively. 
As expected, for estimation of $\beta_1$ that is not confounded,  
both OLS and the null treatments approach  have little bias in both cases even if the number of confounders is not correctly specified.
In Case 1 where the null treatments assumption is met, Null$_1$ has little bias 
because the number of confounders is correctly specified,
and as shown in Table \ref{tbl:cvp2}, the $95\%$  bootstrap confidence interval based on Null$_1$  has coverage probabilities approximate to the nominal level of $0.95$.
But Null$_2$    is biased because   the dimension of the confounder   is specified to be smaller  than the truth.
The bias is smaller than the OLS, although this is not theoretically guaranteed.
If the dimension of the confounder is larger than the truth, the factor analysis fails and so does  the  null treatments estimation.
In Case 2, the null treatments assumption is violated because more than half of the confounded treatments are
active. 
In this case, the null treatments estimation is in general biased for the confounded treatments;
the bias could be   larger or smaller than   OLS.
Therefore, to apply the null treatments estimation, one has to first assure that  the majority of the treatments are null or
very close to zero. 
Both the auxiliary variables and the null treatments approaches rest on correct specification of number of unmeasured  confounders. 
If it is not known with high confidence, we refer to \cite{bai2002determining,chen2012inference,owen2016bi} and \cite{mclachlan2019finite} 
for the estimation methods, and we recommend   assessing robustness of  estimation by varying the specification     in data analyses.

\begin{table}[H]
	\centering
	\caption{Coverage probability of the $95\%$ bootstrap confidence interval for Case 1} \label{tbl:cvp2}
\begin{minipage}{0.75\textwidth}
	\begin{tabular}{cccccccccccc}
 &&$\beta_1$&$\beta_2$&$\beta_3$&$\beta_4$&$\beta_5$&$\beta_6$&$\beta_7$&$\beta_8$\\
\multirow{2}{*}{Null$_1$}
&&0.967& 0.970& 0.973& 0.975& 0.987& 0.974& 0.990& 0.983\\
&&0.943& 0.958& 0.959& 0.967& 0.969& 0.961& 0.963& 0.975
	\end{tabular}
\end{minipage}\hfil
\begin{minipage}{0.23\textwidth}
Note:  The first row is for sample size 2000 and  the second for  5000.
\end{minipage}
\end{table}

\graphicspath{{Results/simunullsetting1/}}

\begin{figure}[H]
	\centering
\includegraphics[scale=0.3]{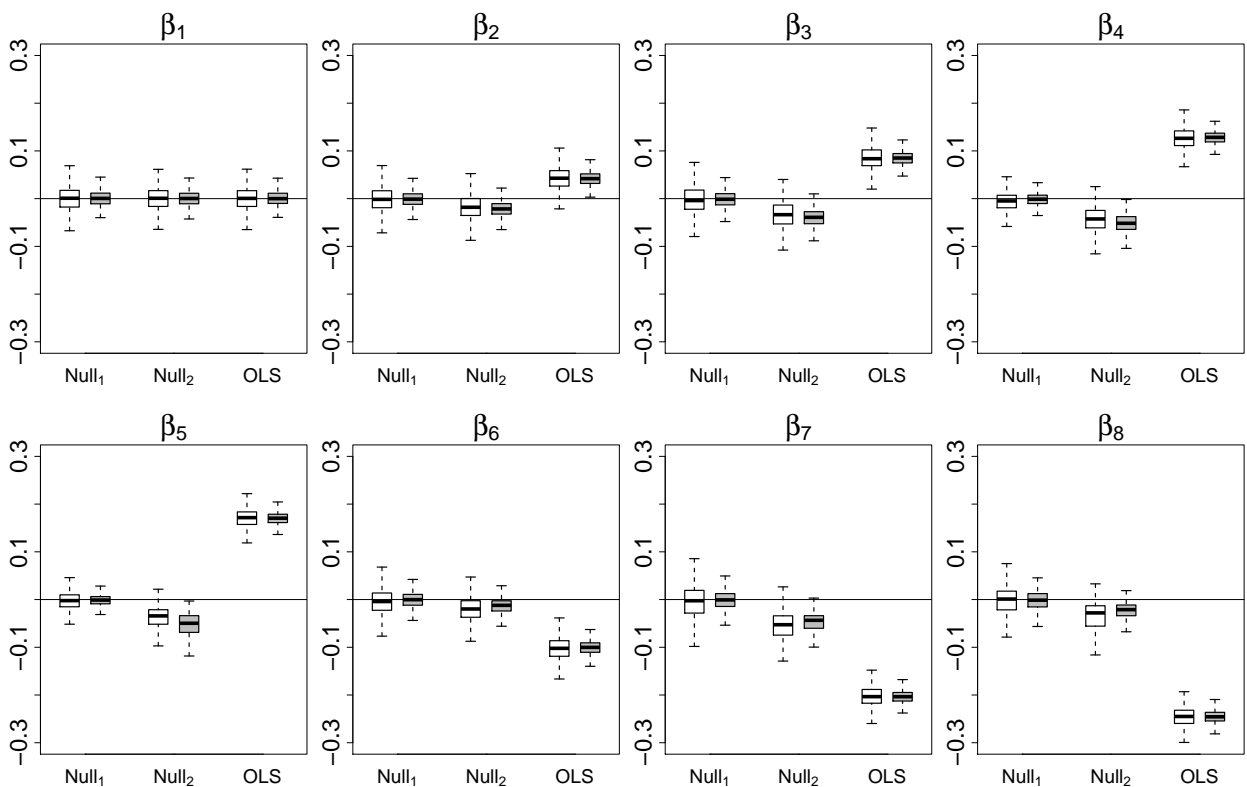}
\caption{Bias of estimators in Case 1.
White boxes are for sample size 2000 and gray ones  for 5000.} \label{fig:bias3}
\end{figure}

\section{Application to a mouse obesity study}\label{sec:application}
For further illustration,  we reanalyze a mouse obesity dataset described by \citet{wang2006genetic},
where  the effect of gene expressions on the body weight of F2 mice is of interest. 
Unmeasured confounding may arise   in such gene expression studies due to   batch effects or unmeasured phenotypes. 
The data we use  are collected from 227 mice, including  the body weight $(Y)$, 17 gene expressions ($X$), and 5 single nucleotide polymorphisms ($Z$); see the supplement for a complete list of these variables. 
As previously selected by \cite{lin2015regularization}, the 17 genes are likely to affect mouse weight and the 5 single nucleotide polymorphisms  are potential instrumental variables.
We further evaluate the effects of these  genes on mouse weight by adopting  a factor model that is widely used to characterize the unmeasured confounding in gene expression studies  \citep{gagnon2012using,wang2017confounder}.
We assume a linear outcome model and estimate the parameters  with three methods:   the auxiliary variables approach  using  single nucleotide polymorphisms as instrumental variables,  the null treatments approach   assuming that fewer than half of genes can affect the mouse weight, and ordinary least squares.
We also compute the bootstrap confidence intervals.

Figure \ref{fig:mice1} presents the point estimates  and  their significance for the 17 genes, 
when the factor number is specified as one for the auxiliary variables and the null treatments estimation. 
Detailed  results including point and interval estimates are relegated to the supplement.
All three methods agree with positive and significant effects of \texttt{Gstm2}, \texttt{Sirpa}, and \texttt{2010002N04Rik}, and a negative effect of  \texttt{Dscam}. 
Additionally, the auxiliary variables estimation  also indicates  negative and significant effects of \texttt{Igfbp2}, \texttt{Avpr1a},   \texttt{Abca8a}, and \texttt{Irx3}; 
the null treatments estimation  suggests  a potentially positive  effect of \texttt{Gpld1}.
These results reinforce previous findings  about the effects of genes on obesity.
For instance, recent studies show that \texttt{Igfbp2}  (Insulin-like growth factor binding protein 2) protects against the development of obesity \citep{wheatcroft2007igf};
\texttt{Gpld1} (Phosphatidylinositol-glycan-specific phospholipase D) is associated with  the change in insulin sensitivity in response to the low-fat diet  \citep{gray2008plasma};
and \texttt{Irx3} (Iroquois homebox gene 3) is associated with  lifestyle changes and plays a crucial role in  determining weight via energy balance \citep{schneeberger2019irx3}.

However,  the significance test \citep[chapter IV]{kim1978factor} of the hypothesis that one factor is  sufficient  is rejected,  indicating that either the number of factors is too small or there exists model misspecification.
Figure S.3 in the supplement shows the results when the number of factors is increased to two and three.
With two factors,
\texttt{Gstm2}, \texttt{2010002N04Rik}, \texttt{Igfbp2}, and \texttt{Avpr1a} remain significant in  the auxiliary variables analysis, and \texttt{Gstm2} and  \texttt{Dscam} remain significant in   the null treatments  analysis.
With three factors,
all estimates in both the    auxiliary variables and the null treatments analyses are no longer significant.
In summary, the association between the 17 genes and mice obesity can be explained by three or more unmeasured confounders. But if there exist only one or two confounders,  \texttt{Gstm2}, \texttt{Sirpa}, \texttt{2010002N04Rik}, \texttt{Dscam}, \texttt{Igfbp2}, \texttt{Avpr1a},   \texttt{Abca8a},  \texttt{Irx3}, and \texttt{Gpld1}   have
a potential causal association with  mouse obesity, which can not be completely attributed to confounding.
In future studies,  it is of interest to identify the potential confounders and to use additional  data to more confidently estimate   effects of these 9 genes.

\graphicspath{{Results/}}
\begin{figure}[H]
\includegraphics[scale=0.36]{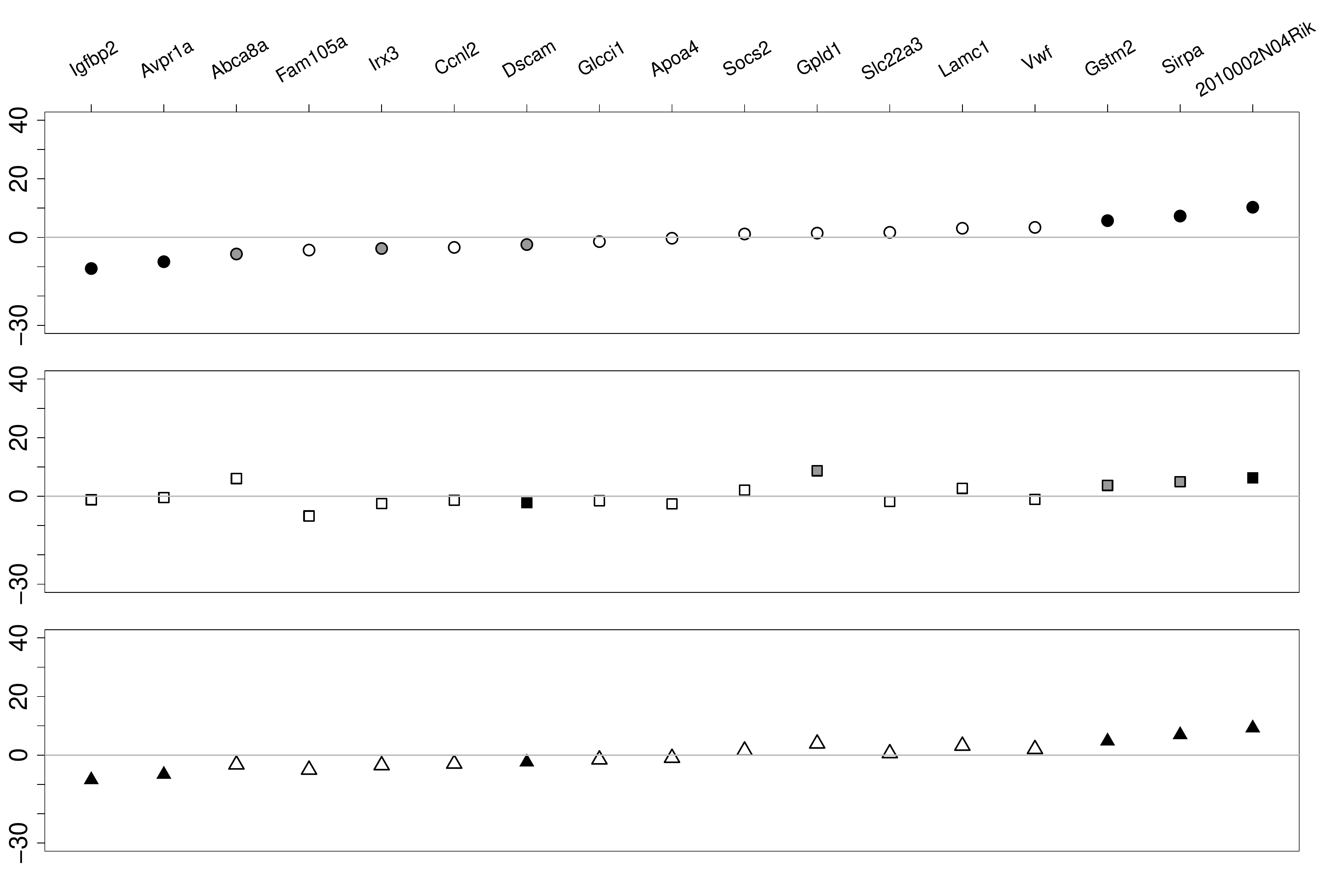}
\caption{Effect estimates for  17 genes when one factor is used in analyses. 
The first panel is for the auxiliary variables estimation, the second for the null treatments estimation, and the last for ordinary least squares estimation. Black points are for  significant estimates  at level of $0.05$, gray ones for $0.1$, and white ones for estimates  not significant at   $0.1$.}\label{fig:mice1}
\end{figure}

\section{Discussion}\label{sec:conclusion}
In this paper, we extend results that had previously been developed for the identification of treatment effects in the presence of unmeasured confounding in the single-treatment to the multi-treatment setting, 
and  we   extend the  parametric approach of \cite{wang2017confounder}  to   identification  of multi-treatment effects with an unrestricted outcome model. 
We demonstrate a novel framework    using integral equations and completeness conditions  to establish identification in causal inference problems.

We have assumed that the number of confounders is known, 
which is realistic  in confounder measurement error or misclassification problems.
When it is not known a priori, consistent estimation of the number of confounders has been well established by \cite{bai2002determining} under factor models and by \cite{chen2012inference} under mixture models.
The \texttt{R} software \texttt{factanal} provides a significance test of whether the number of factors in a factor model is sufficient  to capture the full dimensionality of the data set.
We also refer to  \cite{owen2016bi} and \cite{mclachlan2019finite} for a comprehensive literature review related to choosing the number of confounders in practice.
We also recommend    conducting a sensitivity analysis  like in the simulations and the application by altering the specification of the number of confounders to assess the robustness of the substantive   conclusion.

Our identification framework    rests on the auxiliary variables   or the null treatments assumption.
These assumptions are partially testable:
a heuristic approach,  taking equation (4) as an example, is to check whether a solution exists. This can be  achieved by obtaining  a solution $\tpr(y\mid u,x)$ that minimizes the mean squared error $||\pr(y\mid x,z)-\int_u \tpr(y\mid u,x)\tpr(u\mid x,z)du||^2$ and checking how far away the error is from zero to assess  how well  the solution fits the equation.
This is a typical  goodness-of-fit test  if all models are parametric. 
However,   in nonparametric models   statistical  inference  for the integral equation is quite difficult and we leave the  development of   falsification tests     for future research.
Even if both the auxiliary  variables  and the null treatments assumptions fail to hold,
we describe how to   test whether treatment effects exist.
The proposed estimation strategies   can also be used to test whether unmeasured confounding is present, 
by assessing how far   the proposed estimates   are from the crude ones.

Our  identification results lead to feasible estimation methods under parametric estimation assumptions. 
The proposed estimation methods, comprised of standard factor analysis, linear and robust linear regression, inherit properties from the classical theory of statistical  inference;
these methods work well in simulations and an application. 
However, statistical inference for nonparametric and semiparametric models remains to be studied. 
We have considered fixed dimensions of treatments and confounders, and extensions to large and high-dimensional  settings are of  both theoretical and practical interest.

\section*{Acknowledgements}
We are grateful for   comments from   the editors and three anonymous reviewers.
We thank Eric Tchetgen Tchetgen  for   comments on early versions of this  article.
This work was partially  supported by  Beijing Natural Science Foundation (Z190001) and National Natural Science Foundation of China (12071015 and 12026606), and    ONR grants N00014-18-1-2760 and N00014-21-1-2820.

\section*{Supplementary material}

Supplementary material online includes   
proof of theorems and propositions, useful lemmas, 
discussion and examples on the completeness condition,
consistency of the least median of squares estimator,
discussion on identification of a  parametric model for a binary outcome, 
details for examples,  
additional  results for simulations and the application,
and codes for reproducing the results in this article.

\bibliographystyle{jasa}
\bibliography{CausalMissing}

\appendix
\setcounter{proposition}{0}
\setcounter{lemma}{0}
\setcounter{example}{0}
\renewcommand {\theexample} {S.\arabic{example}}
\renewcommand {\theproposition} {S.\arabic{proposition}}
\renewcommand {\theassumption} {S.\arabic{assumption}}
\renewcommand {\thesection} {S.\arabic{section}}
\renewcommand {\theequation} {S.\arabic{equation}}
\makeatletter

\renewcommand{\@seccntformat}[1]{{\csname the#1\endcsname}.\hspace*{1em}}
\makeatother

\newpage

\setcounter{proposition}{0}
\setcounter{assumption}{0}
\renewcommand {\theexample} {S.\arabic{example}}
\renewcommand {\theproposition} {S.\arabic{proposition}}
\renewcommand {\theassumption} {S.\arabic{assumption}}
\renewcommand {\theequation} {S.\arabic{equation}}
\renewcommand {\thetable} {S.\arabic{table}}
\renewcommand {\thefigure} {S.\arabic{figure}}
\makeatletter

\renewcommand{\@seccntformat}[1]{{\csname the#1\endcsname}.\hspace*{1em}}
\makeatother

 \begin{center}
	\LARGE \bf 
Online supplement to ``Identifying effects of multiple treatments in the presence of unmeasured confounding"
\end{center}

This supplement includes   
\begin{itemize}
\item proof of theorems and propositions, useful lemmas, 
\item discussion and examples on the completeness condition,
\item consistency of the least median of squares estimator,
\item discussion on identification of a  parametric model for a binary outcome, 
\item details for examples,   and
\item  additional  results for simulations and  the application.
\end{itemize}

\section{Proof of propositions  and theorems}
\subsection{Proof of Proposition \ref{prop:factor}}
Note that $\eta$ can be identified by regression of $X$ on $Z$, 
then applying  lemma 5.1 and theorem 5.1 of \cite{anderson1956}  to the factor model  for the residuals,
\begin{eqnarray*}
X-\eta Z = \alpha U + \varepsilon,
\end{eqnarray*}
we obtain   (i) of Proposition \ref{prop:factor}.
The third result of Proposition \ref{prop:factor}  can be obtained from the well-known completeness property of exponential families, see   Theorem 2.2 of \cite{newey2003instrumental}.
Here we prove (ii), which rests on the following lemma  described  by \citet[lemma 1 and remark 5]{kotlarski1967characterizing}.

\begin{lemma}[Kotlarski, 1967]\label{lem:kotlarski}
Let $U$, $\varepsilon_1$, and $\varepsilon_2$ be three independent $q$-dimensional real random vectors with mean zero, 
and let  $W_1= U+\varepsilon_1$ and $W_2=U+\varepsilon_2$.
If the joint characteristic function of  $(W_1,W_2)$ does not vanish, 
then the distributions of $U$, $\varepsilon_1$, and $\varepsilon_2$ are uniquely determined from the joint distribution of $(W_1,W_2)$. 
\end{lemma}

We apply Kotlarski's lemma to prove (ii) of Proposition \ref{prop:factor}.
\begin{proof}[Proof of Proposition \ref{prop:factor} (ii)]
We denote $W=X-\eta Z=\alpha U + \varepsilon$. 
Note that from (i) of proposition  \ref{prop:factor}, any admissible  value for $\alpha$ can be written as $\tilde \alpha =\alpha R$ with $R$ an arbitrary $q\times q$ orthogonal matrix, we only need to prove that  given $\tilde \alpha=\alpha R$, the joint distribution  $\tpr(w,u)=\pr(W=w, U=u; \tilde \alpha)$ is uniquely determined and  $\tpr(w,u)=\pr(W=w, R^\T U=u; \alpha)$.

Because after deleting any row of $\alpha$ there remain two full-rank submatrices of $\alpha$,
there must exist  two disjoint square submatrices of $\alpha$ with full rank $q$.
Note that $\tilde \alpha =\alpha R$, there must exist  two disjoint square submatrices of $\tilde \alpha$ with full rank $q$, which we denote by  $\tilde \alpha_{\mathcal I}$ and $\tilde \alpha_{\mathcal J}$ with ${\mathcal I}$
and ${\mathcal J}$ denoting   the corresponding indices, respectively.
Note that $W_{\mathcal I}= \tilde \alpha_{\mathcal I} V + \varepsilon_{\mathcal I}$ and  $W_{\mathcal J}= \tilde \alpha_{\mathcal J} V + \varepsilon_{\mathcal J}$ 
with $V=R^\T U$,   
we have $\tilde \alpha_{\mathcal I}^{-1} W_{\mathcal I}=  V + \tilde \alpha_{\mathcal I}^{-1} \varepsilon_{\mathcal I}$  and $\tilde \alpha_{\mathcal J}^{-1} W_{\mathcal J}=  V + \tilde \alpha_{\mathcal J}^{-1}\varepsilon_{\mathcal J}$. 
According to Lemma \ref{lem:kotlarski}, the distributions of $V$, $\tilde \alpha_{\mathcal I}^{-1} \varepsilon_{\mathcal I}$, 
and $\tilde \alpha_{\mathcal J}^{-1}\varepsilon_{\mathcal J}$ are uniquely determined given $\tilde \alpha$,
and therefore, the distribution of $\varepsilon = W -\tilde \alpha V$ is uniquely determined.
As a result,  given $\tilde \alpha$,   there is only one admissible joint distribution,  which must be $\tpr(w,u)=\pr(W=w, R^\T U=u\mid \alpha)$.
\end{proof}

\subsection{Proof of Theorem \ref{thm:auxil} }
\begin{proof}
Under  the equivalence (Assumption \ref{assump:auxil} (ii)),  given any admissible joint distribution $\tpr(x,u\mid z)$, there must exist some invertible function $V(U)$ such that $\tpr(x,u\mid z)=\pr\{X=x, V(U)=u\mid z\}$.
Because $V(U)$ is invertible, the ignorability assumption \ref{assump:ign}  ($Y(x)\ind X\mid U$)  implies that $Y(x)\ind X\mid V(U)$,
the exclusion restriction  $Z\ind Y\mid (U,X)$ implies that $Z\ind Y\mid \{X,V(U)\}$, and the completeness  (Assumption \ref{assump:auxil} (iii)) implies that $\tpr(u\mid x,z)$ is also complete in $z$.
Letting $\tpr(y\mid u,x)=\pr\{y\mid V(U)=u,x\}$, then we have that 
\begin{eqnarray}
\pr\{Y(x)=y\} = \int_u  \pr(y\mid u, x)\pr(u) du= \int_u  \tpr(y\mid u, x)\tpr(u) du,  \label{eq:s1}\\
\pr(y\mid x,z)=\int_u \tpr(y\mid u, x) \tpr(u\mid x,z) du,   \label{eq:s2}
\end{eqnarray}
with $ \tpr(u)$ and $\tpr(u\mid x,z)$ derived from $\tpr(x,u\mid z)$.
Because $\tpr(u\mid x,z)$ is complete in $z$, the solution to \eqref{eq:s2} is unique; 
this is because for any candidate solutions $\tpr_1(y\mid u,x)$ and $\tpr_2(y\mid u,x)$ to \eqref{eq:s2},
we must have that $\int_u \{\tpr_1(y\mid u,x)-\tpr_2(y\mid u,x)\}\tpr(u\mid x,z)=0$, which implies that $\tpr_1(y\mid u,x)=\tpr_2(y\mid u,x)$ by the completeness of $\tpr(u\mid x,z)$ in $z$.
Thus, $\tpr(y\mid u,x)$ is uniquely determined from \eqref{eq:s2}, and $\pr\{Y(x)\}$ is identified by plugging in it into \eqref{eq:s1}.
\end{proof}

\subsection{Proof of Theorem \ref{thm:null}}
\begin{proof}
Under  the equivalence (Assumption \ref{assump:null} (ii)),  for any admissible joint distribution  $\tpr(x,u)$ we must have
some invertible function $V(U)$ such that $\tpr(u,x)=\pr\{X=x, V(U)=u\}$. 
Letting $\C=\{i: \pr(u\mid x) \text{ varies with } x_i\}$ and $\tilde \C=\{i: \tpr(u\mid x) \text{ varies with } x_i\}$,
then we must have $\C=\tilde \C$ by noting that $\tpr(u\mid x) = \pr\{V(U)=u\mid x\}$, i.e., $\C$ can be identified from any admissible joint distribution $\tpr(x,u)$.

Because $V(U)$ is invertible, the ignorability assumption \ref{assump:ign}  $(Y(x)\ind X\mid U)$ implies that $Y(x)\ind X\mid V(U)$,
and the completeness  (Assumption \ref{assump:null} (iii)) implies that $\tpr(u\mid x)$ is also complete in $x_\S$ for any $\S\subset \C$ with cardinality $q$.
Letting $\tpr(y\mid u,x)=\pr\{y\mid V(U)=u,x\}$, then we have that 
\begin{eqnarray}
\pr\{Y(x)=y\} = \int_u  \pr(y\mid u, x)\pr(u) du = \int_u  \tpr(y\mid u, x)\tpr(u) du,  \label{eq:s3}\\
\pr(y\mid x)=\int_u \tpr(y\mid u, x) \tpr(u\mid x) du,   \label{eq:s4}
\end{eqnarray}
with $ \tpr(u)$ and $\tpr(u\mid x)$ obtained from $\tpr(x,u)$.

We prove that $\tpr(y\mid u,x)$ is uniquely determined from \eqref{eq:s4} given $\pr(y\mid x)$ and  $\tpr(u\mid x)$ by way of contradiction.
Suppose two candidate outcome models $\tpr_1(y\mid u,x)$ and $\tpr_2(y\mid u,x)$ satisfy \eqref{eq:s4}, then 
$\int_u \{\tpr_1(y\mid u,x)-\tpr_2(y\mid u,x)\} \tpr(u\mid x)=0$. 
Under the null treatments assumption, each of $\tpr_1(y\mid u,x)$ and $\tpr_2(y\mid u,x)$ can depend on only $(|\C|-q)/2$ confounded treatments,
and thus the contrast $\{\tpr_1(y\mid u,x)-\tpr_2(y\mid u,x)\}$ can depend on at most $|\C|-q$ confounded treatments.
We let $X_\S$ denote the rest $q$ confounded treatments that the contrast $\{\tpr_1(y\mid u,x)-\tpr_2(y\mid u,x)\}$ does not depend on,
then the completeness (Assumption \ref{assump:null}(iii)) implies that $\tpr(u\mid x_\S,x_{\bar\S})$ is complete in $x_\S$, 
and thus $\{\tpr_1(y\mid u,x)-\tpr_2(y\mid u,x)\}=0$ almost surely, 
i.e., $\tpr_1(y\mid u,x)=\tpr_2(y\mid u,x)$ almost surely.
Therefore, the solution to \eqref{eq:s4} must be unique. 
Finally, plugging in $\tpr(y\mid u,x)$ and $\tpr(u)$ into \eqref{eq:s3} identifies the potential outcome distribution.
\end{proof}

\subsection{Proof of Proposition \ref{prop:null}}
\begin{proof}
We first note that the  confounded treatments can be identified under the equivalence assumption, by the argument in the proof of Theorem 2.
Note that  the candidate solutions depending  on more than $(|\C|-q)/2$ confounded treatments contradict the null treatments assumption that at most $(|\C|-q)/2$ confounded ones can  affect the outcome,
we only focus on solutions that depends on at most $(|\C|-q)/2$ confounded treatments.

We let $A_\C$ denotes the number of active ones of the confounded treatments.
Consider a solution  $\tpr(y\mid u, x_{\mathcal B})$ that solves
\begin{eqnarray}\label{eq:int4}
\pr(y\mid x) = \int_u \tpr(y\mid u,x_{\mathcal B})\tilde\pr(u\mid x)du.
\end{eqnarray}
where $|\mathcal B \cap \C| \leq (|\C|-q)/2$, i.e.,  $x_{\mathcal B}$ includes at   most  $(|\C|-q)/2$ confounded treatments.
Equation \eqref{eq:int4} is  a Fredholm integral equation of the first kind with the kernel  $\tilde\pr(u\mid x)$
complete in $x_{\bar {\mathcal B}}$, where $x_{\bar {\mathcal B}}$ denotes the remaining treatments of $x$ except for $x_{\mathcal B}$.
For  $x_{\B}$  that includes all active   treatments, i.e.,    $x_\A  \subset x_\B$ and $|\mathcal B \cap \C| =A_\C$, 
the solution to \eqref{eq:int4}
exists and must be unique and equal to $\tpr(y\mid u, x_\A)$.
For $x_{\B}$ that includes   $t < A_\C$  active   ones of the confounded treatments, i.e.,  $|\mathcal B \cap \C| =t<A_\C$, the solution to  \eqref{eq:int4} does not exist. We prove this by way of contradiction.

Suppose  \eqref{eq:int4} has a solution  $\tpr(y\mid u,x_\B)$, then it must  also satisfy the following equation
\begin{eqnarray} \label{eq:int5}
\pr(y\mid x) = \int_u \tpr(y\mid u, x_\B \cup x_\A) \tilde\pr(u\mid x)du;
\end{eqnarray}
where the unknown function $\tpr(y\mid u, x_\B \cup x_\A) $ of $u$  is allowed to 
depend on  all active treatments.
Note that   $\tpr(y\mid u, x_\B \cup x_\A) $ can depend on  at most $ (|\C|-q)/2 -t $ null ones of the confounded treatments, Equation \eqref{eq:int5}  is a Fredholm integral equation of the first kind, 
with  the kernel $\tilde\pr(u\mid x)$ complete in $x_{\bar \B}\cap x_{\bar \A}\cap x_\C$, i.e.,   the remaining $q+t$ null ones  of the confounded treatments. 
Therefore, \eqref{eq:int5} can  be satisfied by  only one function, which in fact is $\tpr(y\mid u,x_\A)$.
This contradicts that $\tpr(y\mid u,x_\B)$ depend on only $t<A_\C$ active ones of the confounded treatments.

As a result, all solutions that depend on at most $(|\C|-q)/2$ confounded treatments must be equal to  $\tpr(y\mid u, x_\A)$, i.e., the solution to \eqref{eq:int2}.

\end{proof}

\subsection{Proof of Theorem \ref{thm:ln}}
We first describe a lemma that is useful for proof of  Theorem \ref{thm:ln}.
\begin{lemma}\label{lem:factor}
For  a $p\times p$ positive-definite matrix  $\Sigma_\varepsilon$  and a $p\times q$  matrix $\alpha$ of full column rank  with $p>q$, letting  $\gamma= (\Sigma_\varepsilon  + \alpha\alpha^\T)^{-1}\alpha$, then 
$\gamma=\Sigma_\varepsilon^{-1} \alpha \{I_q -  \alpha^\T (\Sigma_\varepsilon + \alpha\alpha^\T)^{-1} \alpha\}$.
\end{lemma}
\begin{proof}
Letting $A= \Sigma_\varepsilon^{-1/2}\alpha $ and $B=\Sigma_\varepsilon^{1/2}\gamma$,  then $A$ has full rank, 
and it is straightforward to verify that 
\[B=(I_p + A A^\T)^{-1}A = A\{I_q- A^\T(I_p+AA^\T)^{-1}A\}.\]
Because $\{I_q- A^\T(I_p+AA^\T)^{-1}A\}\{A^\T(I_p+AA^\T)A\} = A^\T A$ and $A^\T A$ has full rank of $q$, 
then $\{I_q- A^\T(I+AA^\T)^{-1}A\}$ must have  full rank of $q$ as a $q\times q$ matrix, i.e., $\{I_q -  \alpha^\T (\Sigma_\varepsilon + \alpha\alpha^\T)^{-1} \alpha\}$ has full rank.
Thus, we have that  
\begin{eqnarray*}
\gamma = \{\Sigma_\varepsilon + \alpha\alpha^\T\}^{-1}\alpha= \Sigma_\varepsilon^{-1} \alpha \{I_q -  \alpha^\T (\Sigma_\varepsilon + \alpha\alpha^\T)^{-1} \alpha\}.
\end{eqnarray*}

In the special case  that $\Sigma_\varepsilon$ is diagonal and $q=1$, we have $ \gamma_i = \Sigma_{\varepsilon,i}^{-1} \alpha_i \{1 -  \alpha^\T (\Sigma_\varepsilon + \alpha\alpha^\T)^{-1} \alpha\}$,
where  $\Sigma_{\varepsilon,i}$ is the $i$th diagonal element of  $\Sigma_\varepsilon$. 
Therefore, if $\alpha_i\neq 0$, we must have $\gamma_i\neq0$.
\end{proof}

We then prove Theorem \ref{thm:ln}.
\begin{proof}[Proof of  Theorem \ref{thm:ln}]
Under model \eqref{mdl:linear} and condition (i) of Theorem \ref{thm:ln}, 
$\Sigma_\varepsilon$ is identified and any admissible value  $\tilde \alpha$ is a rotation of the truth (Proposition \ref{prop:factor}), 
i.e., $\tilde \alpha =\alpha R$ for some  $q\times q$ orthogonal matrix $R$.
We let $ \gamma=\Sigma_X^{-1}\alpha$ denote the coefficient by linear regression of  $U$ on $X$.
Given an admissible value $\tilde \alpha=\alpha R$,  we let $\tilde \gamma=\Sigma_X^{-1}\tilde\alpha=\gamma R$ and $\tilde \delta=R^\T \delta$. We let	 $\C=\{i:\alpha_i \text{ is not a zero vector}\}$  denote the confounded treatments, 
where $\alpha_i$ is the $i$th row of $\alpha$.

Note that $\Sigma_\varepsilon$ is diagonal under model \eqref{mdl:linear},  then according to Lemma \ref{lem:factor}, we have $\tilde \gamma_i = \Sigma_{\varepsilon,i}^{-1} \alpha_i \{I_q -  \alpha^\T (\Sigma_\varepsilon + \alpha\alpha^\T)^{-1} \alpha\}R$,
where $\tilde \gamma_i$ is the $i$th row of $\tilde \gamma$ and $\Sigma_{\varepsilon,i}$ is the $i$th diagonal element of  $\Sigma_\varepsilon$. Therefore,   $\tilde \gamma_i$ is a zero vector if and only if $\alpha_i$ is a zero vector, i.e., $\C$ is identified by the set $\{i: \tilde \gamma_i \text{ is not a zero vector}\}$.

Letting $ \xi$ denote the ordinary least squares coefficient by regressing $Y$ on $X$, then  we have
\begin{equation}\label{eq:s5}
\xi= \beta + \tilde\gamma \tilde \delta.
\end{equation}
By way of contradiction, we prove that $(\beta, \tilde \delta)$ is uniquely determined from this equation  given $(\tilde \beta, \tilde \gamma)$ and under the null treatments assumption.
Suppose that two sets of values $(\beta^{(1)},\tilde \delta^{(1)})$ and $(\beta^{(2)},\tilde \delta^{(2)})$
satisfy \eqref{eq:s5}, and both $\beta_\C^{(1)}$ and $\beta_\C^{(2)}$ satisfies the null treatments assumption that   at most $(|\C|-q)/2$  entries are nonzero.

For unconfounded treatments, the corresponding rows of $\tilde \alpha$ and $\tilde \gamma$ must be zero as we show in the above,    thus $\beta_{\bar\C}^{(1)}=\beta_{\bar\C}^{(2)}=\tilde \beta_{\bar\C}$.
We remain to prove that $\beta_{\C}^{(1)}=\beta_{\C}^{(2)}$ and $\tilde \delta^{(1)}=\tilde \delta^{(2)}$.
From \eqref{eq:s5},  we have that  
\[\beta_\C^{(1)}  - \beta_\C^{(2)}  = \tilde \gamma_\C  (\tilde \delta^{(2)} - \tilde \delta^{(1)} ).\]
On the left hand side of this equation,  $\beta_\C^{(1)}  - \beta_\C^{(2)}$ has at least $q$ zero entries under  assumption  (i) of Theorem \ref{thm:ln}.
We use $\mathcal Z$ to denote the  indices of zero entries of $\beta_\C^{(1)}  - \beta_\C^{(2)}$ and $\tilde\gamma_{\mathcal Z}$ the corresponding submatrix of $\tilde\gamma_\C$,
then we have that $0=\tilde \gamma_{\mathcal Z} (\tilde \delta^{(2)} - \tilde \delta^{(1)} )$.
Note that $\tilde \gamma_{\mathcal Z}=\gamma_{\mathcal Z} R$   must have full rank of $q$ under assumption (iii) of Theorem \ref{thm:ln},
then we must have $\tilde \delta^{(2)} =\tilde \delta^{(1)}$, and thus $\beta_\C^{(1)}  = \beta_\C^{(2)}$.
In summary, $\beta^{(1)}=\beta^{(2)}$, i.e., $\beta$ is uniquely determined.
\end{proof}

\section{Discussion and examples on the completeness condition}
Completeness is a fundamental   concept   in statistics (see \cite{lehman1950completeness,basu1955statistics}), 
which is  taught in most foundational courses of statistical inference.
It has been used to establish the theory for hypothesis testing and unbiased estimation in mathematical statistics \citep{lehman1950completeness}, and recently been used to establish identification in causal inference, missing data, and measurement error problems.
Nonetheless, it may  still be  abstract to practitioners.
Therefore, it is worth explaining in more detail. 
We add further explanation and extra examples to facilitate the interpretation and use  of the completeness condition in practice.
In particular,  we illustrate   completeness from the following perspectives.

\begin{itemize}
\item The role of completeness in identification. 
Since its  prevalent use in statistics, completeness    has been  widely used to  establish   identification  for  a variety of nonparametric and semiparametric models,  for instance, the   IV regression model \citep{newey2003instrumental,darolles2011nonparametric},
IV quantile model \citep{chernozhukov2005iv,chen2014local}, measurement error model \citep{hu2008instrumental}, missing data model \citep{miao2016varieties,d2010new}, and proximal inference \citep{miao2018proxy}. 
It has been very well studied by statisticians and economists,  primitive conditions  are readily available in the literature including  very general exponential families of distributions and regression models, and the literature is still growing; see for example,  \cite{newey2003instrumental,d2010new,d2011completeness,darolles2011nonparametric,chen2014local,hu2018nonparametric}. 
Our use of completeness falls in this line of  work, where the main identifying assumption that captures the underlying causal structure  is  an IV, auxiliary variables, or  null treatments assumption and the  completeness  is viewed as a regularity condition.

\item  Intuition  and implication of  completeness. 
Completeness is equivalent to the injectivity of the conditional expectation operator \citep{d2011completeness}.
Completeness characterizes the informativeness of the auxiliary variable about the confounder and its ability to recover the confounding bias.
It is analogous to the relevance   condition in the  instrumental variable identification.
It can be interpreted as a nonparametric rank condition and  is easiest understood in the  categorical  and the linear cases where the outcome model  to be identified is parametric.
In the categorical case where  both $U$ and $Z$ have $k$ levels, completeness means that  the matrix $[\pr(u_i\mid x,z_j)]$ consisting of the conditional probabilities is invertible for any $x$. 
This is stronger than dependence of $Z$ and $U$ given $X$.
Roughly speaking, dependence reveals that variability in $U$ is accompanied by  variability in $Z$, 
and  completeness reinforces that  any infinitesimal  variability in $U$ is accompanied by   variability in $Z$.
For instance,  if $Z$ is a proxy of $U$, completeness of $\pr(u\mid x,z)$ can be interpreted as no coarsening in the  measurement $Z$ of the confounder $U$.
As a consequence, completeness fails if the number of levels or dimension of $Z$ is smaller  than that of $U$.
For the binary   case, completeness holds if   $U$ and $Z$ are correlated within each level of $X$.
In the linear model  $E(U\mid x,Z)=\gamma_0(x)+\gamma_1(x)Z$, completeness reduces to a   rank condition that $\gamma_1(x)$  has full row rank for all $x$. 
The rank condition can only hold if the dimension of $Z$ is no smaller than  that of $U$.
This argument   provides a rationale for measuring a rich set of potential auxiliary variables for the purpose of confounding adjustment.
However,  if the outcome model is unrestricted,    completeness serves as a   generic   rank condition accommodating  both categorical and continuous variables, linear and nonlinear models,
although, it can no longer be expressed so   concisely as a full rank condition.

\item How to assess or test completeness.
Completeness can be checked in specific models, for instance, one can check whether the covariance matrix is of full rank in the joint normal model. 
Unfortunately, \cite{canay2013testability} show that  for unrestricted models the  completeness condition  is in fact untestable, even if all relevant variables ($X,Z,U$ in our problem) are observed.
Therefore, without restrictions on the distribution, it is    impossible to provide empirical evidence in favor of the completeness condition, akin to  the ignorability assumption.

\item When   does completeness hold or fail,  and is it a stringent condition?
A number of papers \citep{andrews2017examples,d2011completeness,newey2003instrumental,darolles2011nonparametric,chen2014local,hu2018nonparametric} have established  genericity results for parametric, semiparametric, and nonparametric distributions satisfying completeness. 
\cite{andrews2017examples} has  shown that if $Z$ and $U$ are continuously distributed and the dimension of $Z$ is larger than that of $U$,  then  under    a mild regularity condition  the completeness condition holds generically in the sense that   the set of distributions or conditional expectation operators for which completeness fails has a property analogous to having zero Lebesgue  measure \citep{chen2014local,andrews2017examples}.
By appealing to such results,   completeness holds in a large class of distributions and thus one may argue that it is  commonly satisfied.

%In the context of our application, ??? 

\end{itemize}

In short, completeness is one of the most general conditions made in problems of identification.
It requires  that $Z$ must have sufficient dimensions or levels and variability relative to   $U$.
Commonly-used parametric and semiparametric models, such as exponential families \citep[Theorem 2.2]{newey2003instrumental} and location-scale families \citep{mattner1992completeness,hu2018nonparametric}, and nonparametric additive models \citep{d2011completeness} satisfy the completeness condition.
For nonparametric models, it is not testable but  holds in a large class of models.

In the following, we  provide extra examples   illustrating completeness, see also \cite{lehman1950completeness} for a variety of parametric examples where completeness holds and also counterexamples.
We also refer to  \cite{newey2003instrumental} for completeness of  exponential families, \cite{hu2018nonparametric} for location-scale families,   and \cite{d2011completeness,darolles2011nonparametric}   for     additive separable regression models.

\begin{example}
The  binary case.
Suppose both $Z$ and $U$ are binary, then for any $x$ completeness of $\pr(u\mid x,z)$ holds as long as $U\nind Z\mid X=x$, but otherwise completeness fails  if $U\ind Z\mid X=x$.
\end{example}

\begin{example}
The  categorical case.
Suppose  $U$ has $q$ levels and $Z$ has $r$ levels, 
then for a given  $x$ completeness of $\pr(u\mid x,z)$ in $z$ holds as long as the matrix 
\[[\pr(u_i\mid x,z_j)]_{q\times r}=\left\{\begin{array}{ccc}
{\pr(u_1\mid z_1,x)}&\cdots& {\pr(u_1\mid z_r,x)}\\
\colon&\ddots&\colon\\
{\pr(u_q\mid z_1,x)}&\cdots& {\pr(u_q\mid z_r,x)}
\end{array}
\right\}\]
consisting of the conditional probabilities  has full row rank.
Therefore, it is necessary that $q\leq r$ and $U\nind Z\mid X=x$. 
Otherwise, completeness fails if either $q> r$ or $U\ind Z\mid X=x$.
However for  $q>2$, the full rank condition  is stronger than the dependence  ($U\nind Z\mid X=x$).
For instance, if $\pr(u_1\mid z,x)\neq \pr(u_2\mid z,x)$ and $\pr(u_3\mid z,x)= \pr(u_1\mid z,x)$  for all $z$,
then $U\nind Z\mid X=x$ but the full rank condition is obviously not met. 
This is because the variability in $U$ from $u_1$ to $u_3$ is not sufficiently captured by $Z$, i.e., 
the measure of $Z$ is coarsened if we view it as a proxy of $U$. 
Roughly speaking, dependence reveals that variability in $U$ is accompanied by  variability in $Z$, 
and  completeness reinforces that  any infinitesimal  variability in $U$ is accompanied by   variability in $Z$.
\end{example}

\begin{example}
Gaussian distributions.
Suppose $U$ and $Z$ have dimensions of $q$ and $r$, respectively, and  $\pr(u,z\mid x)$ is joint normal given $x$, then  completeness  of $\pr(u\mid x,z)$ in $z$ reduces to a   rank condition: the coefficient matrix  $\gamma_1(x)$ in model $E(U\mid x,Z)=\gamma_0(x)+\gamma_1(x)Z$ has full row rank   given  $x$. 
It is required that   the dimension of $Z$ is no smaller than that of $U$ and that the regression coefficients 
of each confounder on $X$ and $Z$ are not collinear; otherwise, the completeness fails. 
\end{example}

\begin{example}
A scale model.
\citet[example 3.3]{lehman1950completeness} presents a counterexample where completeness fails
for  $\pr(u\mid x,z) \thicksim  N(0, \sigma_{x,z}^2)$.
This is because the conditional density is  an even function of $u$  and $E\{g(U)\mid x,z\}=0$ for any square-integrable and odd function $g$.
In this example, the scale or magnitude of variability of $U$ is captured by $Z$  but not the orientation.

\end{example}

\section{Consistency of the least median of squares estimator  $(\tilde \delta^{\lowercase{\rm lms}},\hat\beta^{\rm \lowercase{lms}})$}
For the consistency of $(\hat \delta^{\rm lms}, \hat\beta^{\rm lms})$, we need an additional regularity condition that is slightly stronger than assumption (i) of Theorem \ref{thm:ln}, which is routinely assumed  in the  least median squares estimation \citep[see Theorem 3 in Chapter 3 of][]{rousseeuw2005robust}.
\begin{assumption}\label{assump:s1}
At most  $[|\C|/2]-q+1$ entries of $\beta_\C$  are nonzero, where $[x]$ is  the largest integer less than   or equal to $x$.
\end{assumption}

We show consistency of $(\hat \delta^{\rm lms}, \hat\beta^{\rm lms})$ under this assumption and the assumptions of Theorem \ref{thm:ln}, 
given   $n^{1/2}$-consistency of $(\hat\xi,\hat\gamma)$, i.e.,  $n^{1/2}(\hat\xi -\xi)$ and   $n^{1/2}(\hat\gamma - \gamma R)$ are bounded in probability for some unknown orthogonal matrix $R$. 
We show that $\hat \delta^{\rm lms}\rightarrow R\delta$ and $\hat\beta^{\rm lms}\rightarrow \beta$.
For notational simplicity, we only consider the special case where $R$ is  the identity matrix. 
For general cases with  $R$  unknown, the following proof holds by simply replacing $\delta$ with $R\delta$ and $\gamma$ with $\gamma R$.

Because $n^{1/2}(\hat\gamma - \gamma)$ is bounded in probability,
then $||\hat \gamma_i||_2^2 \rightarrow || \gamma_i||_2^2 > \log(n)/n$ for $\gamma_i\neq 0$ 
and $n/ \log(n) ||\hat \gamma_i||_2^2\rightarrow 0$ for $\gamma_i=0$. 
Lemma \ref{lem:factor} implies $\C=\{i: ||\alpha_i||_2^2 > 0\} = \{i: || \gamma_i||_2^2 > 0\}$
and therefore, $\pr(\hat \C\neq \C) \rightarrow 0$, i.e., $\hat \C$ consistently selects the confounded treatments.  
Letting
\begin{eqnarray*}
\tilde  \delta^{\rm lms} =\arg\min_{\delta}  \text{median}\  \{ (\hat \xi_i  - \hat \gamma_i \delta)^2,\  i\in   \C\},\quad    \C=\{i: || \gamma_i||_2^2 > 0\},
\end{eqnarray*}
we only need to show consistency of $\tilde  \delta^{\rm lms}$ because  $\tilde  \delta^{\rm lms}= \hat  \delta^{\rm lms}$ upon $\hat\C = \C$.

Note that 
\begin{eqnarray*}
\med\{ (\hat\xi_i  - \hat \gamma_i \tilde \delta^{\rm lms})^2: i \in \C\} 
&\leq&  \med\{ (\hat\xi_i  - \hat \gamma_i \delta)^2: i \in \C\} \\
&\leq& \med \{( \hat \xi_i - \xi_i  - (\hat\gamma_i - \gamma_i)\delta + \xi_i - \gamma_i\delta)^2: i \in \C\}\\
&\leq& \med \{( \hat \xi_i - \xi_i  - (\hat\gamma_i - \gamma_i)\delta +\beta_i)^2: i \in \C\}.
\end{eqnarray*}
For sufficiently large sample size $n$, $\hat\xi_i -\xi_i$ and $\hat\gamma_i - \gamma_i$ are close to zero  so that
$( \hat \xi_i - \xi_i  - (\hat\gamma_i - \gamma_i)\delta)^2 < ( \hat \xi_j - \xi_j  - (\hat\gamma_j - \gamma_j)\delta +\beta_j)^2$ for any $i$ with $\beta_i=0$ and any $j$ with $\beta_j \neq 0$.
Assumption \ref{assump:s1}  states  that more than  half entries of $\beta_\C$ are zero, and thus, 
$\med \{( \hat \xi_i - \xi_i  - (\hat\gamma_i - \gamma_i)\delta +\beta_i)^2: i \in \C\}$ is attained among the null treatments. 
Therefore, we have asymptotically 
\begin{eqnarray*}
\med \{( \hat \xi_i - \xi_i  - (\hat\gamma_i - \gamma_i)\delta +\beta_i)^2: i \in \C\}
&\leq & \max \{( \hat \xi_i - \xi_i  - (\hat\gamma_i - \gamma_i)\delta)^2: i\in \C\text{ and } \beta_i=0\}\\
&\leq&  \max \{( \hat \xi_i - \xi_i  - (\hat\gamma_i - \gamma_i)\delta)^2: i\in \C\}.
\end{eqnarray*}
Hence,
\begin{eqnarray*}
\med\{ (\hat\xi_i  - \hat \gamma_i \tilde \delta^{\rm lms})^2: i \in \C\} &\leq&  \max \{( \hat \xi_i - \xi_i  - (\hat\gamma_i - \gamma_i)\delta)^2: i\in \C\}.
\end{eqnarray*}
Letting $\Delta = \tilde \delta_{\rm lms} - \delta$, we can show the following  result,  
\begin{gather*}
\text{\bf Result 1:\quad }\med\{ (\hat\xi_i  - \hat \gamma_i \tilde\delta^{\rm lms})^2: i\in \C\} \geq   \frac{1}{2}\{(\hat\xi_i - \xi_i) - (\hat\gamma_i - \gamma_i)(\delta+\Delta) - \gamma_i\Delta\}^2\\
\text{for   at least $q$ elements belonging to the subset $\{i\in \C:\beta_i=0\}$}.
\end{gather*}
Given Result 1, we have 
\begin{gather*}
\frac{1}{2}(\hat\xi_i - \xi_i - (\hat\gamma_i - \gamma_i)(\delta+\Delta) - \gamma_i\Delta)^2 \leq  \max \{( \hat \xi_i - \xi_i  - (\hat\gamma_i - \gamma_i)\delta)^2: i\in \C\}\\
\text{for   at least $q$ elements belonging to the subset $\{i\in \C:\beta_i=0\}$}.
\end{gather*}
Assuming that $(\hat\xi,\hat\gamma)$ are consistent,  then the right hand side must converge to zero and  thus $\gamma_i\Delta \rightarrow 0$ for at least $q$ elements in $\C$.  Moreover,  any submatrix of $\gamma_\C$ consisting of $q$ rows has full rank (Assumption (iii) of Theorem 3), then $\Delta$ must converge to zero, i.e., $\tilde\delta_{\rm lms}$ is consistent and as a result $\hat\delta_{\rm lms}$ is consistent. 
Consistency of $\hat \beta^{\rm lms}=\hat\xi - \hat\gamma \hat\delta^{\rm lms}$ follows from consistency 
of $(\hat\xi,\hat\gamma,\hat\delta_{\rm lms})$.

Now we prove Result 1. 
Note that
\begin{eqnarray*}
\med\{ (\hat\xi_i  - \hat \gamma_i \tilde\delta^{\rm lms})^2: i\in \C\} 
&=& \med \{(\beta_i + (\hat\xi_i - \xi_i) - (\hat\gamma_i - \gamma_i)(\delta+\Delta) - \gamma_i\Delta)^2: i\in \C\}.
\end{eqnarray*}
If $|\C|$ is odd, Assumption \ref{assump:s1} implies that at most $(|\C| - 1)/2 - q + 1$  entries of $\beta_{\C}$ are nonzero. 
Arranging $(\hat\xi_i  - \hat \gamma_i \tilde \delta^{\rm lms})^2$  in increasing order, 
then $\med\{ (\hat\xi_i  - \hat \gamma_i \tilde\delta^{\rm lms})^2: i\in \C\} $ is equal to the $(|\C|+1)/2$-th element. 
Thus, the following inequality holds for  at least $(|\C| +1 )/2 - \{ (|\C| - 1)/2 - q+1\} = q$  elements  belonging to the subset $\{i\in \C: \beta_i=0\}$,
\begin{eqnarray*}
\med\{ (\hat\xi_i  - \hat \gamma_i \tilde \delta^{\rm lms})^2: i \in \C\} &\geq& \{(\hat\xi_i - \xi_i) - (\hat\gamma_i - \gamma_i)(\delta+\Delta) - \gamma_i\Delta\}^2   \\ 
 &\geq& \frac{1}{2}\{(\hat\xi_i - \xi_i) - (\hat\gamma_i - \gamma_i)(\delta+\Delta) - \gamma_i\Delta\}^2.
\end{eqnarray*}
If $|\C|$ is even, Assumption \ref{assump:s1} implies that at most $|\C|/2 - q + 1$  entries of $\beta_{\C}$ are nonzero.  
Arranging $(\hat\xi_i  - \hat \gamma_i \tilde \delta^{\rm lms})^2$  in increasing order, 
then $\med\{ (\hat\xi_i  - \hat \gamma_i \tilde\delta^{\rm lms})^2: i\in \C\}$ is equal to the average of the $|\C|/2$-th and $(|\C|/2+1)$-th  elements.
Thus,  $\med\{ (\hat\xi_i  - \hat \gamma_i \tilde\delta^{\rm lms})^2: i\in \C\}$ is no smaller  than half of the $(|\C|/2+1)$-th element.  
As a result,  the following inequality holds for at least $|\C|/2 + 1 - (|\C|/2 - q +1) = q$  elements  belonging to the subset $\{i\in \C: \beta_i=0\}$,
\begin{eqnarray*}
\med\{ (\hat\xi_i  - \hat \gamma_i \tilde \delta^{\rm lms})^2: i \in \C\} 
 &\geq& \frac{1}{2}\{(\hat\xi_i - \xi_i) - (\hat\gamma_i - \gamma_i)(\delta+\Delta) - \gamma_i\Delta\}^2.
\end{eqnarray*}
This completes the proof of Result 1.

\section{Discussion on  identification of a  parametric binary outcome model}

Without assist of  auxiliary variables and  null treatments assumptions,  identification is not generally available and depends on specific  model assumptions.
In a recent note,  \cite{kong2019multi} consider a    binary outcome model with  one confounder.
Under  a factor model for     normally distributed treatments and   a couple of  assumptions such as knowing the sign of confounding bias, they  prove  identification  via a meticulous analysis  of the link distribution.
However, their identification results do not generalize to the multivariate confounder case as illustrated by the following counterexample.

\begin{example}

Assuming that
\begin{eqnarray}\label{mdl:glm}
X= \alpha U + \varepsilon,\quad 
\pr(Y=1\mid X,U) = G(\beta_0 + \beta^\T X + \delta^\T U),
\quad \varepsilon\thicksim N(0, \Sigma_\varepsilon),\quad  U\thicksim N(0, I_q),
\end{eqnarray}
where $U$ is a $q$-dimensional confounder,  $\Sigma_\varepsilon$ is diagonal, and  $G$ is a known  distribution    function relating the outcome mean to a linear model of the treatments and confounder.
%A special case is the   logistic model considered by \cite{kong2019multi}.
The unknown parameters $(\beta,\delta)$ capture the treatment effects and the magnitude of confounding, respectively. 

Under this setting, one can verify that  the observed data distribution $\pr(x,y)$ is satisfied with 
$\tilde \alpha =\alpha R_1$,  $\tilde \delta = R_1^\T \Sigma^{-1/2}R_2 \Sigma^{1/2} \delta$, and $\tilde \beta = \beta + \gamma(I_q - \Sigma^{-1/2}R_2 \Sigma^{1/2})\delta$,
where  $\Sigma=I_q - \alpha^\T \Sigma_X^{-1}\alpha$ and $R_1,R_2$ are arbitrary $q\times q$ orthogonal matrices.

In the special case where $U$ is univariate,  i.e., $q=1$, there are only two possible values $-1$ and $1$  for orthogonal matrices $R_1,R_2$; thus,  there are at least  two possible values for the treatment effect,   $\tilde \beta=\beta$  and $\tilde \beta=\beta + 2\gamma\delta$. 
If further the signs of $\delta$ and at least one entry of $\alpha$ are known, i.e., $R_1=R_2=1$, then $\tilde \beta=\beta + 2\gamma\delta$ can be excluded, and in fact,  \cite{kong2019multi} have shown that $\tilde \beta=\beta$ is the only possible value for the treatment effect provided that $G$ is not a normal distribution.
However, this argument does not generalize to the multivariate confounder case, because there are infinite number of orthogonal matrices with dimension $q\geq 2$, in which case, it is impossible to specify $R_1,R_2$.
\end{example}

\section{Details for examples}
\subsection{Details for Example \ref{examp:iv}}
Note that $\eta$ can be identified by regression of $X$ on $Z$.
Given  $\eta$,  an arbitrary admissible value $\tilde \alpha$, and $\tpr(u\mid x,z) \thicksim N(\tilde \gamma^\T x - \tilde \gamma^\T  \eta z, \tilde \sigma^2)$ with $\tilde \gamma= (\Sigma_{X-\eta Z})^{-1}\tilde \alpha$ and $\tilde \sigma^2 = 1 - \tilde \alpha ^\T (\Sigma_{X-\eta Z})^{-1}\tilde \alpha$, we solve
\begin{eqnarray}\label{eq:integral}
\pr(y\mid x,z) &=& \int_u h(y,x,u)\tpr(u\mid x,z)du\\
&&\text{$\phi$ is the probability density function of $N(0,1)$,}\\
&=& \int_u h(y,x,u)\cdot \frac{1}{\tilde\sigma}\phi\left\{\frac{u - (\tilde \gamma^\T x - \tilde \gamma^\T  \eta z)}{\tilde\sigma}\right\}du
\end{eqnarray} 
for $h(y,x, u)$, which is the outcome model $\tpr(y\mid x,u)$.
Following the procedure described by \cite{miao2018proxy}, 
$h(y,x,u)$ can be represented in Fourier transforms of  $\tpr(u\mid x,z)$ and $\pr(y\mid x,z)$.

By  substitution  $z'=\{\tilde \gamma^\T x -  \tilde \gamma^\T\eta z\}/\tilde \sigma$, $u'=u/\tilde \sigma$,  and by letting 
\[g(y, x, z')=\pr\left\{y\mid x, z= \frac{\tilde \gamma^\T x - \tilde \sigma z' }{\tilde \gamma^\T \eta}\right\},\]
\eqref{eq:integral} implies that 
\begin{eqnarray*}
g(y, x, z')&=&\int_{-\infty}^{+\infty}  \frac{1}{\tilde\sigma} \phi(z' - u')\cdot  h(y,x, u'\tilde \sigma) du\\
&=&\int_{-\infty}^{+\infty}    \phi(z' - u')\cdot   h(y,x, u'\tilde \sigma) du'.
\end{eqnarray*}
which is an integral equation of convolution type and  can be solved  by applying the Fourier transform.
Letting $h_1$ and $h_2$ denote the Fourier transforms of  $\phi$ and $g$  respectively:
\begin{align*}
h_1(t)&=\int_{-\infty}^{+\infty}\exp(-{\rm i}tz)\phi(z)dz,\\
h_2(y,x, t)&=\int_{-\infty}^{+\infty} \exp(-{\rm i}t z') g(y,x,z')dz'\\
&=- \frac{\tilde \gamma^\T \eta}{\tilde \sigma}\int_{-\infty}^{+\infty} \exp\left\{-{\rm i}t \frac{\tilde \gamma^\T x - \tilde \gamma^\T \eta z}{\tilde \sigma}\right\} \pr(y\mid x,z)dz,
\end{align*}
with  ${\rm i}=(-1)^{1/2}$ the imaginary unity, we have
\[h_2(y,x, t)=h_1(t)\times \int_{-\infty}^{+\infty}\exp(-\i t u') h(y,x, u'\tilde \sigma) du',\]
\[\int_{-\infty}^{+\infty}\exp(-\i t u')h(y,x, u'\tilde \sigma) du'=\frac{h_2(y,x, t)}{h_1(t)};\]
by  Fourier inversion,  we have
\[h(y,x, u'\tilde \sigma)=\frac{1}{2\pi}\int_{-\infty}^{+\infty}\exp(\i t u')\frac{h_2(y,x, t)}{h_1(t)} dt;\]
by  substitution $u=u'\tilde \sigma$, we obtain
\[h(y,x, u)=\frac{1}{2\pi}\int_{-\infty}^{+\infty}\exp\left\{\frac{\i t u}{\tilde \sigma}\right\}\frac{h_2(y,x, t)}{h_1(t)} dt,\]
and the potential outcome distribution is 
\[\pr\{Y(x)=y\} = \int_{-\infty}^{+\infty} h(y,x, u)\phi(u)du.\]

\subsection{A confounder proxy example for the auxiliary variables approach}
For simplicity, we consider a binary confounder, $p$ binary treatments, and a binary proxy $Z$ of the confounder.
We assume that  
\begin{gather}
X_1\ind\cdots \ind X_p\mid U,\label{eq:factor2}\\ 
\text{at least three treatments are correlated with $U$;}\label{eq:factor3}\\
 Z\nind U,  \quad Z\ind (X,Y)\mid U; \label{eq:factor4}
\end{gather} 
where the last  independence is known as the  nondifferentially  error assumption \citep{ogburn2013bias,carroll2006measurement}. 
Under \eqref{eq:factor2}--\eqref{eq:factor4}, completeness of $\pr(u\mid x, z)$ in $z$ holds as long as $Z$ is correlated with $U$, because 
$E\{g(U)\mid x,z\}=0\Leftrightarrow \sum_u g(u)\pr(x,u)\pr(z\mid u)=0\Leftrightarrow g(u)\pr(x,u)=0\Leftrightarrow g(u)=0$.

According to  \cite{kuroki2014measurement},   under \eqref{eq:factor2}--\eqref{eq:factor3}, 
any admissible joint distribution $\tilde \pr(x,u)$ equals the joint distribution of $X$ and  some label switching   of $U$.
Given $\tilde\pr(x,u)$,  we  solve $\pr(z,x)=\sum_u \tpr(z\mid u)\tilde \pr(x,u)$ to obtain $\tilde \pr(z\mid u)$ and  $\tilde \pr(x,z,u)=\tilde \pr(x,u)\tilde \pr(z\mid u)$.
We then obtain $\tilde \pr(y\mid u,x)$ by solving $\pr(y\mid x,z)=\sum_u  \tpr(y\mid u, x)\tpr(u\mid x, z)$,  and finally the potential outcome distribution is identified by $\pr\{Y(x)=y\}=\sum_u \tpr(y\mid u,x)\tpr(u)$.
%%\[\frac{\{\pr(v_1) - \pr(v_1\mid x, z_0)\}\pr(y\mid x, z_1)}{\pr(v_1\mid x, z_1) - \pr(v_1\mid x, z_0)} + 
%%\frac{\{\pr(v_1\mid x, z_1) - \pr(v_1)\}\pr(y\mid x, z_0)}{\pr(v_1\mid x, z_1) - \pr(v_1\mid x, z_0)}.\]
%\[\frac{\{\tilde \pr(u) - \tilde \pr(u\mid x, z=0)\}\pr(y\mid x, z=1)}{\tilde \pr(u\mid x, z=1) - \tilde \pr(u\mid x, z=0)} + 
%\frac{\{\tilde \pr(u\mid x, z=1) - \tilde \pr(u)\}\pr(y\mid x, z=0)}{\tilde \pr(u\mid x, z=1) - \tilde \pr(u\mid x, z=0)}.\]
The identification  result can be generalized to the categorical setting, 
and we refer to \cite{kuroki2014measurement}  for  details of factor analysis in this case.

\subsection{A normal mixture model that satisfies the equivalence assumption}

\begin{example} \citep[Proposition 2]{yakowitz1968identifiability}.
Suppose $U$ has $q$ categories with $\pr(U=u_i)=\pi_i$ and $X$ is a $p$ dimensional vector with $\pr(x\mid u_i)\thicksim N(\mu_i,\Sigma_i)$.  
Assuming that   the pairs $(\mu_i,\Sigma_i)$ are all distinct,  
then $\pr(x)$ has a unique representation in normal mixtures $\pr(x) \thicksim  \sum_{i=1}^{q'} p_i' N(\mu_i',\Sigma_i')$:
we must have $q'=q$ and for each $i$ there must exist some $j$ such that $\pi_i'=\pi_j$ and $(\mu_i',\Sigma_i') = (\mu_j,\Sigma_j)$. 
That is, the equivalence holds and $\pr(x,u)$ is identified up to a label switching of the confounder.

\end{example}

\section{Additional results for simulations and the  application}
\subsection{Results for simulations}

\graphicspath{{Results/simuauxsetting2/}}
\begin{figure}[H]
	\centering
\includegraphics[scale=0.3]{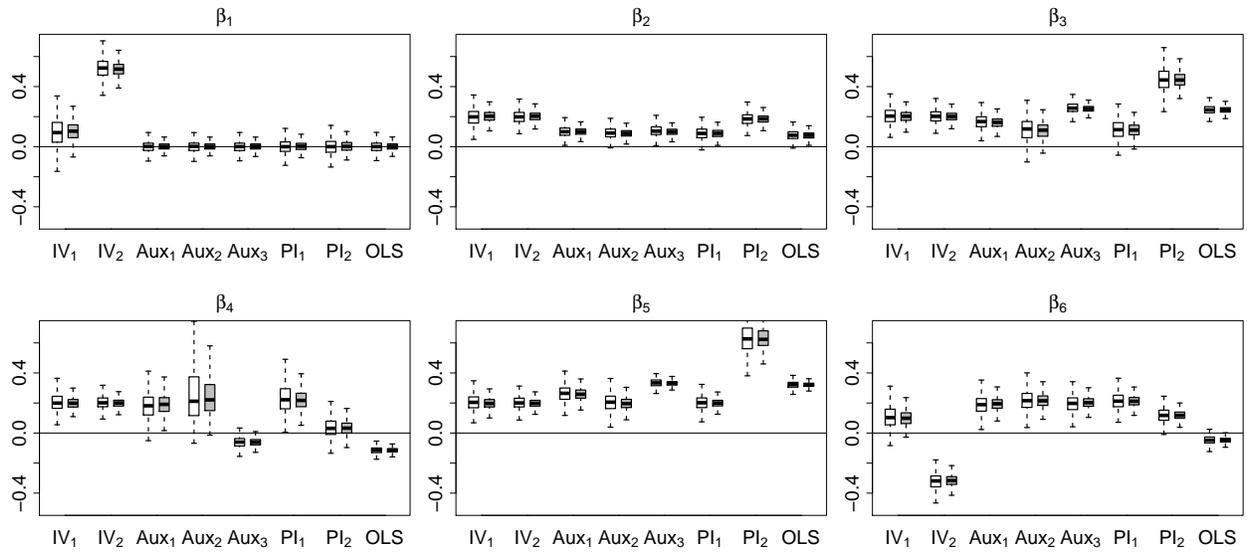} 
\caption{Bias of estimators when  the exclusion restriction fails in the auxiliary variables setting.
White boxes are for sample size 1000 and gray ones  for 2000.} \label{fig:bias2}
\end{figure}

\graphicspath{{Results/simunullsetting2/}}

\begin{figure}[H]
	\centering
 
\includegraphics[scale=0.3]{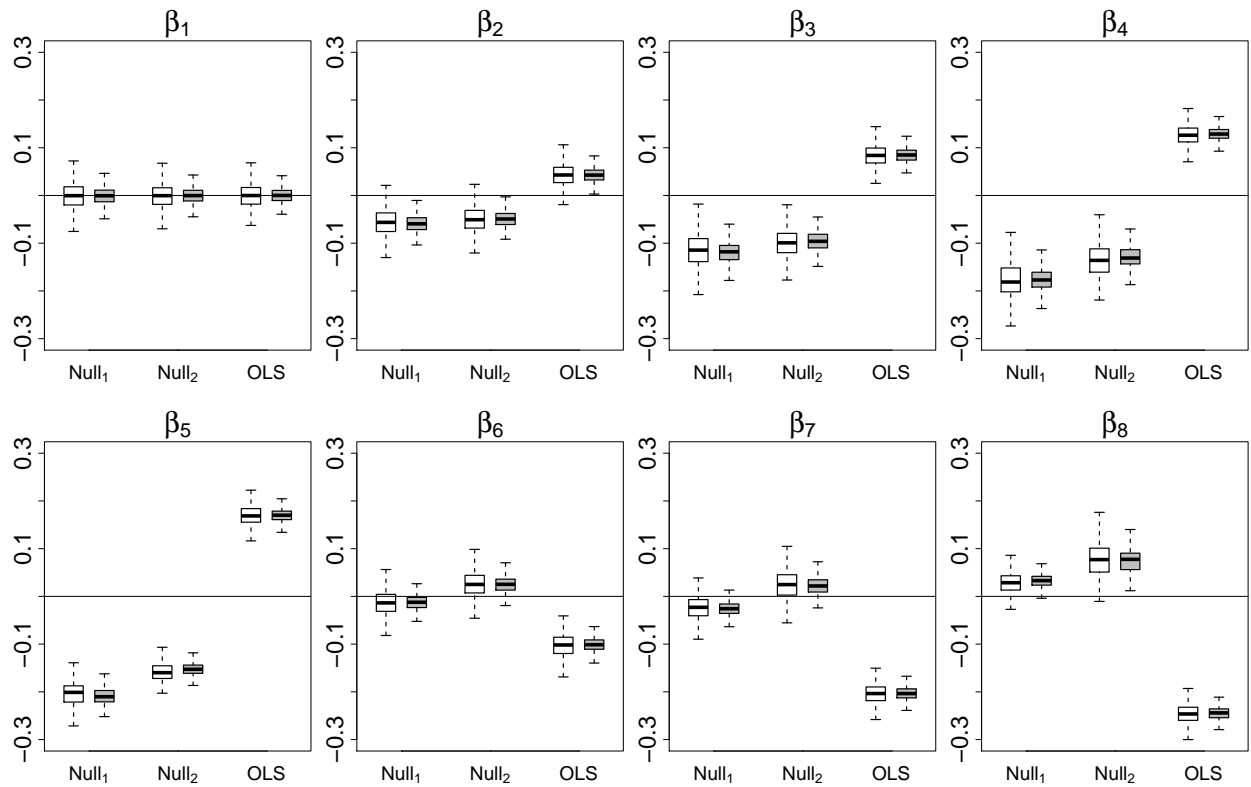}
\caption{Bias of estimators in Case 2 of the null treatments setting. 
White boxes are for sample size 2000 and gray ones  for 5000.} \label{fig:bias4}
\end{figure}

\subsection{Results for the application}

\graphicspath{{Results/}}
\begin{figure}[H]
\includegraphics[scale=0.36]{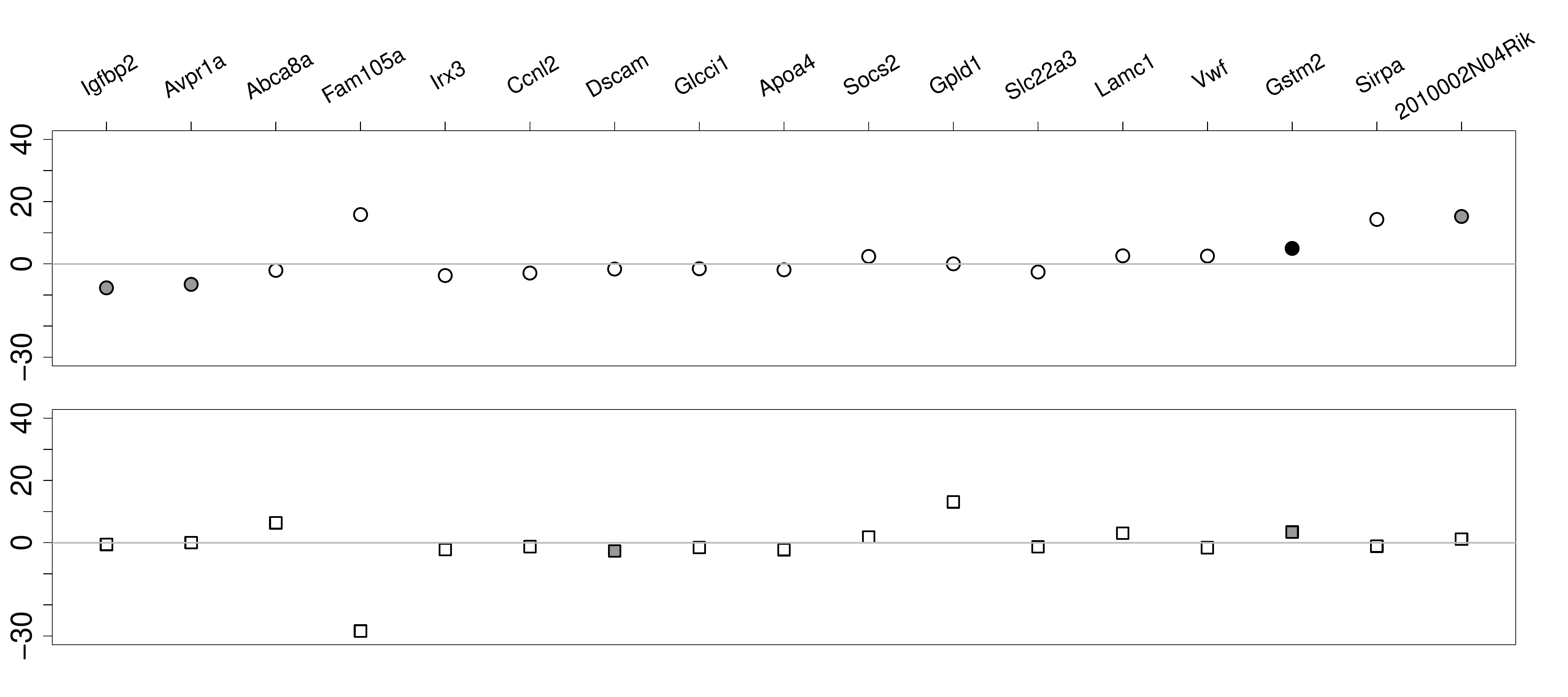}
\includegraphics[scale=0.36]{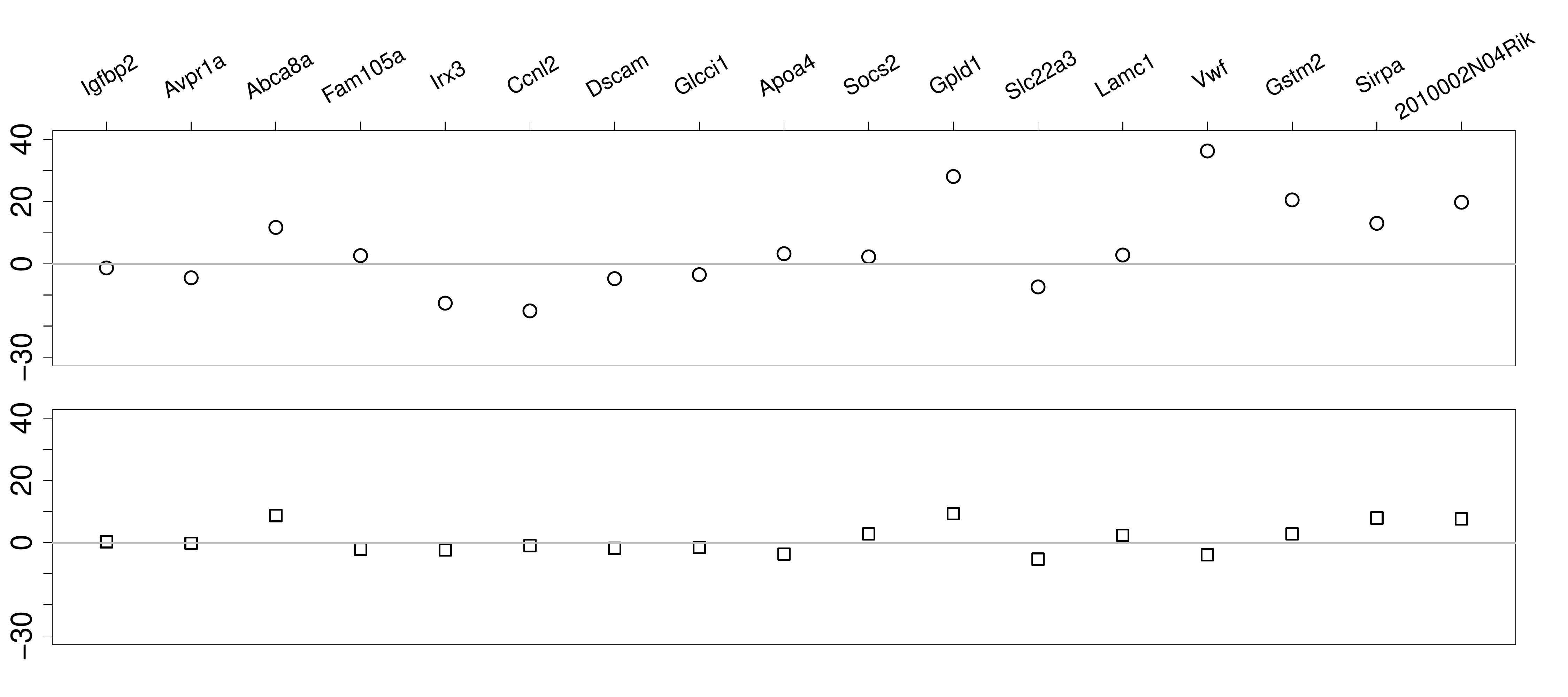}
\caption{Effect estimates for  17 genes when two (the first two panels) or three (the last two panels) factors are used in analyses. 
The first and third panels are for the auxiliary variables estimation, and the second and fourth for the null treatments estimation. Black points are for  significant estimates  at level of $0.05$, gray ones for $0.1$, and white ones for estimates  not significant at   $0.1$.}\label{fig:mice2}
\end{figure}

\begin{center}
Point and confidence interval estimates
\end{center}

Bootstrap percentiles: 2.5\%,  97.5\%,      5\%,    95\%;
Significance codes: ``**" for significant at level of $0.05$, ``*" for $0.1$, and ``0" for not significant at level $0.1$.
\begin{center}
Estimation with one confounder
\end{center}
\begin{verbatim}
1. Results for the auxiliary variables approach 

            estimates    2.5%  97.5%      5%    95% significance
Igfbp2        -10.626 -15.614 -5.266 -14.764 -5.719     **
Avpr1a         -8.296 -14.568 -1.544 -13.494 -2.670     **
Abca8a         -5.664 -12.368  1.072 -10.887 -0.032      *
Fam105a        -4.325 -13.812  4.983 -12.137  2.933      0
Irx3           -3.793  -8.033  0.593  -7.297 -0.154      *
Ccnl2          -3.421  -7.810  0.917  -6.966  0.246      0
Dscam          -2.442  -6.703  0.276  -5.072 -0.514      *
Glcci1         -1.429  -6.999  3.139  -5.964  2.360      0
Apoa4          -0.287  -4.226  3.821  -3.569  3.299      0
Socs2           1.162  -1.800  3.909  -1.276  3.516      0
Gpld1           1.456  -6.520 10.827  -4.847  8.951      0
Slc22a3         1.724  -4.953  7.274  -3.915  6.177      0
Lamc1           3.111  -1.267  8.967  -0.661  8.204      0
Vwf             3.410  -2.669  9.452  -1.400  8.589      0
Gstm2           5.717   1.181 10.045   1.864  9.304     **
Sirpa           7.283   1.312 13.743   2.263 12.774     **
2010002N04Rik  10.293   3.012 17.203   4.351 16.110     **


2. Results for the null treatments approach 

           estimates    2.5%  97.5%      5%    95% significance
Igfbp2        -1.194  -9.369  2.800  -8.045  1.840      0
Avpr1a        -0.430  -7.574  3.999  -6.821  2.678      0
Abca8a         6.032  -3.043 13.693  -1.842 12.352      0
Fam105a       -6.726 -15.875  3.118 -14.526  1.234      0
Irx3          -2.465  -7.292  1.475  -6.661  0.988      0
Ccnl2         -1.349  -6.031  2.855  -5.338  2.083      0
Dscam         -2.255  -7.043 -0.145  -5.540 -0.390     **
Glcci1        -1.536  -6.826  3.343  -6.151  2.567      0
Apoa4         -2.595  -6.378  2.136  -5.713  1.285      0
Socs2          2.117  -1.100  4.992  -0.533  4.513      0
Gpld1          8.681  -0.709 16.796   0.703 14.594      *
Slc22a3       -1.803  -7.358  4.809  -6.227  3.932      0
Lamc1          2.664  -1.792  8.506  -1.235  7.696      0
Vwf           -1.072  -5.372  6.179  -4.309  5.445      0
Gstm2          3.688  -0.806  9.612   0.171  8.948      *
Sirpa          4.972  -0.704 11.731   0.285 10.628      *
2010002N04Rik  6.272   0.096 13.682   1.353 12.599     **


3. Results for the crude estimation 

           estimates    2.5%  97.5%      5%    95% significance
Igfbp2        -8.445 -12.797 -4.586 -12.124 -5.292     **
Avpr1a        -6.607 -11.779 -1.634 -11.094 -2.395     **
Abca8a        -3.236  -8.644  2.049  -7.674  1.429      0
Fam105a       -5.017 -14.030  3.880 -12.662  2.108      0
Irx3          -3.494  -8.028  0.878  -7.284  0.263      0
Ccnl2         -2.962  -7.357  0.860  -6.716  0.269      0
Dscam         -2.410  -7.427 -0.220  -5.806 -0.640     **
Glcci1        -1.522  -7.485  3.178  -6.389  2.303      0
Apoa4         -0.973  -4.563  2.935  -3.992  2.045      0
Socs2          1.518  -1.109  4.289  -0.743  3.828      0
Gpld1          3.919  -2.720 10.907  -1.583  9.980      0
Slc22a3        0.609  -5.184  6.260  -4.205  4.983      0
Lamc1          3.174  -1.191  9.643  -0.598  8.337      0
Vwf            2.032  -3.505  7.896  -2.520  6.781      0
Gstm2          4.767   0.251  9.300   1.194  8.647     **
Sirpa          6.935   0.858 13.027   1.636 12.190     **
2010002N04Rik  9.226   2.454 15.538   3.235 14.574     **
\end{verbatim}

 \begin{center}
Estimation with two confounders 
\end{center}
\begin{verbatim}
1. Results for the auxiliary variables approach 

           estimates    2.5%  97.5%      5%    95% significance
Igfbp2        -7.707 -18.416  1.113 -16.599 -0.978      *
Avpr1a        -6.586 -15.364  0.383 -13.920 -1.049      *
Abca8a        -2.100 -15.114  8.107 -12.496  5.177      0
Fam105a       15.842 -43.428 59.749 -34.904 49.521      0
Irx3          -3.752  -8.337  0.682  -7.363  0.020      0
Ccnl2         -2.959  -7.808  1.329  -7.060  0.594      0
Dscam         -1.671  -6.691  1.455  -5.357  0.861      0
Glcci1        -1.570  -7.136  3.210  -5.886  2.482      0
Apoa4         -1.913  -7.064  4.566  -5.579  3.704      0
Socs2          2.382  -2.566  6.707  -1.939  5.493      0
Gpld1          0.004  -9.004 11.469  -6.927  9.564      0
Slc22a3       -2.617 -15.306 11.915 -11.466  9.435      0
Lamc1          2.585  -1.511  9.001  -1.086  8.108      0
Vwf            2.503  -3.897  9.630  -2.735  8.523      0
Gstm2          4.962   0.284 10.218   1.271  9.119     **
Sirpa         14.276  -6.498 30.434  -2.951 26.550      0
2010002N04Rik 15.232  -2.550 31.459   0.999 25.530      *


2. Results for the null treatments approach 

            estimates    2.5%  97.5%      5%    95% significance
Igfbp2         -0.584   -9.401  2.190  -7.986  1.276      0
Avpr1a          0.021   -8.330  2.998  -7.046  2.156      0
Abca8a          6.359   -3.449 13.216  -2.241 11.147      0
Fam105a       -28.398 -318.341  1.951 -73.696  0.404      0
Irx3           -2.211   -7.370  1.553  -6.523  1.022      0
Ccnl2          -1.270   -6.192  2.666  -5.442  2.121      0
Dscam          -2.692   -7.175  0.208  -5.986 -0.262      *
Glcci1         -1.582   -6.935  3.447  -5.972  2.602      0
Apoa4          -2.275   -6.315  2.239  -5.450  1.679      0
Socs2           1.784   -1.301  5.008  -0.813  4.409      0
Gpld1          13.077   -0.993 19.818  -0.040 17.053      0
Slc22a3        -1.340   -9.239  5.508  -7.669  4.482      0
Lamc1           3.058   -1.893  8.703  -1.245  7.451      0
Vwf            -1.619   -5.267  6.335  -4.201  5.598      0
Gstm2           3.426   -0.295  9.968   0.438  9.129      *
Sirpa          -1.141   -6.001 14.311  -3.833 12.380      0
2010002N04Rik   1.116   -2.534 14.654  -1.160 13.157      0
\end{verbatim}

\begin{center}
Estimation with three confounders 
\end{center}
\begin{verbatim}
1. Results for the auxiliary variables approach 

            estimates    2.5%  97.5%      5%    95% significance
Igfbp2         -1.312 -17.987  6.862 -15.951  4.212      0
Avpr1a         -4.481 -18.904  5.022 -15.353  2.897      0
Abca8a         11.713 -17.207 34.518 -12.953 28.557      0
Fam105a         2.648 -71.350 81.166 -49.248 62.241      0
Irx3          -12.628 -33.805  4.268 -24.106  1.251      0
Ccnl2         -15.119 -23.126  4.215 -18.271  1.616      0
Dscam          -4.751 -15.444  1.729 -11.027  0.785      0
Glcci1         -3.485  -9.837  4.176  -7.731  3.013      0
Apoa4           3.268  -7.404 10.427  -5.519  7.904      0
Socs2           2.249  -5.602  8.652  -3.800  6.912      0
Gpld1          28.073 -50.673 59.428 -19.273 42.508      0
Slc22a3        -7.413 -23.712 12.482 -18.348  9.988      0
Lamc1           2.841  -2.863 13.409  -1.291 10.471      0
Vwf            36.274 -11.822 56.152  -4.906 43.160      0
Gstm2          20.551  -4.118 30.792  -0.593 25.547      0
Sirpa          13.029  -6.271 33.161  -3.244 26.817      0
2010002N04Rik  19.805  -6.478 69.222  -0.918 53.812      0


2. Results for the null treatments approach 

            estimates    2.5%  97.5%      5%    95% significance
Igfbp2         0.304    -9.903   1.264   -8.630  0.689      0
Avpr1a        -0.188    -8.836   3.601   -7.821  2.650      0
Abca8a         8.709    -4.579  13.954   -3.275 10.415      0
Fam105a       -2.115 -1527.892  42.311 -900.908  6.914      0
Irx3          -2.364    -7.044   4.262   -6.313  2.759      0
Ccnl2         -0.936    -5.863   4.479   -5.029  3.555      0
Dscam         -1.812    -6.532   1.339   -5.410  0.149      0
Glcci1        -1.523    -6.817   3.630   -5.850  2.827      0
Apoa4         -3.665    -7.561   2.568   -6.635  1.677      0
Socs2          2.812    -1.284   5.211   -0.670  4.579      0
Gpld1          9.310    -4.096  19.869   -2.003 15.753      0
Slc22a3       -5.310    -7.903   5.516   -6.701  4.519      0
Lamc1          2.353    -2.483   9.319   -1.486  7.890      0
Vwf           -3.877   -15.989   5.917  -11.843  4.612      0
Gstm2          2.807    -3.831   9.466   -2.007  8.653      0
Sirpa          7.923    -4.153  13.763   -2.283 12.411      0
2010002N04Rik  7.609   -19.287 178.871   -2.461 14.933      0
\end{verbatim}

\end{document}